\documentclass[11pt]{amsart}

\usepackage{graphics}
\usepackage{epstopdf}
\usepackage{mathtools}
\usepackage{amsthm,amsmath,amssymb}
\usepackage[numbers,sort&compress]{natbib}
\usepackage{tabularx}
\usepackage{setspace}
\usepackage{booktabs}
\usepackage[margin=3cm]{geometry}
\usepackage{hyperref}
\usepackage{tgtermes}

\usepackage[T1]{fontenc}

\newtheorem{theorem}{Theorem}[section]

\newtheorem{lemma}[theorem]{Lemma}
\newtheorem{corollary}[theorem]{Corollary}
\theoremstyle{definition}
\newtheorem{definition}[theorem]{Definition}

\theoremstyle{remark}

\newtheorem{remark}[theorem]{Remark}

\usepackage{enumerate}
\usepackage{doi,url}

\usepackage{accents}
\newcommand{\wunderbar}[1]{\underaccent{\bar}{#1}}

\usepackage{subfig}
\usepackage{float}

\mathtoolsset{showonlyrefs}

\usepackage{pgfplots}
\usepackage{tikz}
\usetikzlibrary{arrows,positioning}
\tikzset{
    >=stealth',
    pil/.style={
           ->,
           thick,
           auto,
           shorten <=-4pt,
           shorten >=-4pt,}
}

\usepackage{color}
\definecolor{uowyellow}{RGB}{236,183,49}
\definecolor{uowblue}  {RGB}{36,83,151}
\definecolor{uoworange}{RGB}{207,67,56}
\definecolor{uowlblue}  {RGB}{65,149,211}
\definecolor{uowdgreen} {RGB}{55,140,81}
\definecolor{uowlgreen} {RGB}{178,187,56}
\definecolor{uowlorange}{RGB}{217,122,45}
\definecolor{uowred}   {RGB}{168,25,51}
\definecolor{uowpink}   {RGB}{202,0,151}
\definecolor{uowpurple}   {RGB}{67,51,147}
\definecolor{uowgrey}   {RGB}{76,86,95}

\hypersetup{colorlinks,breaklinks,
            linkcolor=uowblue,urlcolor=uowpurple,
            anchorcolor=black,citecolor=uowpink}

\renewcommand{\d}{\mathrm d}
\renewcommand{\O}{\mathcal O}

\newcommand{\Der}[2]{\dfrac{\d {#1}}{\d {#2}}}

\newcommand{\R}{\mathbb{R}}

\newcommand{\re}{\mbox{e}}

\newcommand{\cm}{\mathrm{cm}}
\newcommand{\M}{\mathrm{M}}
\newcommand{\s}{\mathrm{s}}
\newcommand{\cells}{\mathrm{cells}}
\renewcommand{\H}{\mathrm{H}^+}
\newcommand{\ubra}[3]{\underbrace{#1}_{\parbox{#2}{\setstretch{1.0}\centering{\scriptsize{#3}}}}}

\newcommand{\diag}{\mathrm{diag}}
\newcommand{\inte}{\mathrm{int}}
\newcommand{\inner}[2]{\left\langle #1, #2 \right\rangle}

\newcommand{\T}{\mathrm{T}}

\newcommand{\bu}{\mathbf{u}}
\newcommand{\bv}{\mathbf{v}}

\newcommand{\bx}{\mathbf{x}}
\newcommand{\by}{\mathbf{y}}
\newcommand{\bz}{\mathbf{z}}

\newcommand{\bau}{\bar{u}}
\newcommand{\ubau}{\wunderbar{u}}

\newcommand{\bax}{\bar{x}}
\newcommand{\ubax}{\wunderbar{x}}

\newcommand{\baf}{\bar{F}}
\newcommand{\ubaf}{\wunderbar{F}}

\newcommand{\bbu}{\mathbf{\bar{u}}}
\newcommand{\bbv}{\mathbf{\bar{v}}}

\newcommand{\bbx}{\mathbf{\bar{x}}}

\newcommand{\bbF}{\mathbf{\bar{F}}}
\newcommand{\bbG}{\mathbf{\bar{G}}}

\newcommand{\ubu}{\mathbf{\wunderbar{u}}}

\newcommand{\ubx}{\mathbf{\wunderbar{x}}}

\newcommand{\ubF}{\mathbf{\wunderbar{F}}}
\newcommand{\ubG}{\mathbf{\wunderbar{G}}}

\newcommand{\bF}{\mathbf{F}}
\newcommand{\bG}{\mathbf{G}}

\newcommand{\bn}{\mathbf{n}}
\newcommand{\bzero}{{\boldsymbol{0}}}
\newcommand{\bfeta}{{\boldsymbol{\eta}}}

\newcommand{\bfxi}{{\boldsymbol{\xi}}}
\newcommand{\bfzeta}{{\boldsymbol{\zeta}}}
\newcommand{\bfgamma}{{\boldsymbol{\gamma}}}

\providecommand*\emaillink[1]{\nolinkurl{#1}}
\renewcommand{\email}[1]{\href{mailto:#1}{\emaillink{#1}}}

\title{Model for Acid-Mediated Tumour Invasion with Chemotherapy Intervention I: Homogeneous Populations}

\author{Andrew~B. Holder$^\dagger$}
\thanks{$^\dagger$ Corresponding author, \email{aholder@uow.edu.au}}
\author{Marianito~R. Rodrigo$^\ddagger$}
\thanks{$^\ddagger$ \email{marianito\_rodrigo@uow.edu.au}}
\address{School of Mathematics and Applied Statistics, University of Wollongong, NSW 2522, Australia}



\begin{document}

\begin{abstract}
The acid-mediation hypothesis, that is, the hypothesis that acid produced by tumours, as a result of aerobic glycolysis, provides a mechanism for invasion, has so far been considered as a relatively closed system. The focus has mainly been on the dynamics of the tumour, normal-tissue, acid and possibly some other bodily components, without considering the effect of an external intervention such as a cytotoxic treatment. This article aims to examine the effect that a cytotoxic treatment has on a tumour growing under the acid-mediation hypothesis by using a simple set of ordinary differential equations that consider the interaction between normal-tissue, tumour-tissue, acid and chemotherapy drug.
\end{abstract}
\maketitle
\pagestyle{myheadings}
\thispagestyle{plain}
\markboth{\textsc{A.~B. Holder \and M.~R. Rodrigo}}{\textsc{Model for acid-mediated tumour invasion with chemotherapy intervention}}

\section{Introduction}

This article considers the acid-mediation hypothesis with the added interaction of a tumour treatment protocol. The acid-mediation hypothesis is the assumption that tumour invasion is facilitated by acidification of the region around the tumour-host interface caused by aerobic glycolysis, also known as the Warburg effect~\citep{Warburg1926}. This acidification creates an inhospitable environment and results in the destruction of the normal-tissue ahead of the acid resistant tumour thus enabling the tumour to invade into the vacant region. This hypothesis was first examined by Gatenby~and~Gawlinski~\cite{GatenbyGawlinski1996} with a system of reaction-diffusion equations that considers the interaction between the tumour, host and acid. This article examines the acid-mediation hypothesis with the inclusion of population competition as considered in \citep{McGillenetal2013} and also the effect of tumour treatment from a cytotoxic agent such as used for chemotherapy. This will be considered here in a homogeneous environment to gain an understanding of the reaction dynamics that could predict behaviour of an arguably more realistic heterogeneous setting. The heterogeneous setting will be considered in a following article that will utilise a system of reaction-diffusion equations similar to those considered in \cite{GatenbyGawlinski1996,GatenbyGawlinski2003,McGillenetal2013}.

The effect of chemotherapy treatment has yet to be considered in a model that utilises the acid-mediation hypothesis. We wish to present a model that addresses this unexamined question of the interaction of the low extracellular pH of the tumour micro-environment and a cytotoxic tumour treatment. There are however many models that consider chemotherapy and the corresponding effect on the growth of solid tumours. Continuum models have been used in which the dynamics of total cell populations and average chemotherapy drug concentration are considered by employing the use of ordinary differential equations (ODEs), some examples include~\citep{dePillisetal2007,dePillisetal2006,Byrne2003}. There are recent models that consider the addition of an immune response in a tumour cell and chemotherapy model~\citep{dePillisetal2006,dePillisetal2009} encouraged by experimental results suggesting an important impact of the host immune response on the effectiveness of a chemotherapy treatment. Gatenby~and~Gillies~\cite{GatenbyGillies2004} note that highly acidic tumours have been shown to be resistent to anthracyclines as a result of greater phenotypic diversity~\citep{FearonVogelstein1990} which is enabled by mutagenic/clastogenic effects of acidosis. The effects of normal cell populations in a model that considers chemotherapy have largely been neglected. Hence it is an aim of this article is to determine whether the presence of normal cells can alter the perceived effectiveness of chemotherapy.

The article is organised in the following manner. Section~\ref{sec:ModelForm} describes the assumptions made by the model and provides the formulation of the mathematical model being considered. In Section~\ref{sec:sys-constant} the results are presented of a steady-state analysis for the model when treatment characterised by a constant infusion of the chemotherapy drug is considered. The analysis of the model considering regularly scheduled treatments occurring in cycles is presented in Section~\ref{sec:sys-periodic}. A discussion of the results of the analysis of the model considering treatment cycles is given in Section~\ref{sec:discussion}. Concluding remarks have been provided in Section~\ref{sec:concluding-remarks}. Additional results and some of the more laborious calculations required for Sections~\ref{sec:sys-constant}~and~\ref{sec:sys-periodic} have been provided in Appendices~\ref{sec:appendix}--\ref{sec:appendix3}.

\section{Model formulation}\label{sec:ModelForm}

The basic assumptions taken into account to develop the model are
\begin{enumerate}[(i)]
\item Both normal and tumour cells are governed by logistic growth in the absence of any kind of intervention~\citep{GatenbyGawlinski1996,dePillisRadunskaya2003,dePillisetal2006};
\item A population competition relationship exists between the normal and tumour cells~\citep{McGillenetal2013};
\item The tumour-tissue produces $\H$ ions as a result of aerobic glycolysis~\citep{GatenbyGawlinski1996,McGillenetal2013} at a rate proportional to a function of the tumour cell density;
\item The normal-tissue interacts with the excess $\H$ ions, leading to a death rate proportional to the concentration
of $\H$ ions~\citep{McGillenetal2013,GatenbyGawlinski1996};
\item The excess $\H$ ions are produced at a rate proportional to the neoplastic cell density and an uptake term is included to take
account of mechanisms for increasing pH~\citep{GatenbyGawlinski1996};
\item The chemotherapy drug is infused at a rate given by a function of time. A term is included for removal of drug from the system by metabolic processes~\citep{Byrne2003,dePillisetal2006};
\item The tumour-tissue interacts with the chemotherapy drug leading to destruction of tumour-tissue at a rate proportional to the concentration of drug~\citep{Byrne2003,dePillisetal2006};
\item The chemotherapy drug concentration is decreased as a result of interaction with the tumour-tissue~\citep{Byrne2003}.
\end{enumerate}

Let the populations at time $s$ (in $\s$) be denoted by:
\begin{itemize}
\item $N_1(s)$, normal cell density (in $\cells\,\cm^{-3}$),
\item $N_2(s)$, tumour cell density (in $\cells\,\cm^{-3}$),
\item $H(s)$, excess $\H$ ion concentration (in $\M$),
\item $C(s)$, chemotherapy drug concentration (in $\M$).
\end{itemize}
Consider the following model
\begin{align}
\Der{N_1}{s}=&\,\ubra{r_1 N_1 \left(1-\frac{N_1}{K_1}-\alpha_1 \frac{N_2}{K_2}\right)}{12em}{logistic growth with cellular competition}-\ubra{d_1 H N_1}{6em}{normal cell death by acid},\label{eqn:N1}\\
\Der{N_2}{s}=&\,\ubra{r_2 N_2 \left(1-\frac{N_2}{K_2}-\alpha_2 \frac{N_1}{K_1}\right)}{12em}{logistic growth with cellular competition}- \ubra{d_2 C N_2}{6em}{tumour death by drug},\label{eqn:N2}\\
\Der{H}{s}=&\,\ubra{r_3 f(N_2)}{6em}{acid production} - \ubra{m_3 H}{6em}{acid uptake},\label{eqn:H}\\
\Der{C}{s}=&\,\ubra{r_I(s)}{6em}{drug infusion}- \ubra{m_4 C}{6em}{drug decomposition} -\ubra{d_4 C N_2}{6em}{drug-tumour interaction removal}.\label{eqn:A}
\end{align}
The conventions used here are that the subscript for each parameter corresponds to the relevant equation; $r$ represents growth rate; $K$ represents carrying capacity; $\alpha$ represents population competition strength; $d$ represents rate of decrease due to interaction; $m$ represents decrease through system mechanisms. The parameters used in the model, their interpretation and potential values/range of values have been provided in Table~\ref{tab:parameters}.

\begin{table}[htbp]
   \centering
        \caption{Table of parameters and estimated values}
     \begin{tabularx}{\textwidth}{llXll}
     \toprule
     Parameter     & Units         & Description   & Value         & Source \\
     \midrule
     $r_1$         & $\s^{-1}$     & normal cell growth rate & $\O(10^{-6})$ & \citep{GatenbyGawlinski1996,dePillisetal2006} \\
     $r_2$         & $\s^{-1}$     & tumour cell growth rate & $\O(10^{-6})$ & \citep{GatenbyGawlinski1996,dePillisetal2006} \\
     $r_3$         & $\M \, \cm^3\,\s^{-1} \cells^{-1}$  & $\H$ ion production rate & $2.2\times 10^{-17}$ & \citep{MartinJain1994} \\
     $d_1$         & $\M^{-1}\,\s^{-1}$ & fractional normal cell kill by $\H$ ions & $\O(1)$       & \citep{GatenbyGawlinski1996} \\
     $d_2$       & $\M^{-1}\,\s^{-1}$ & fractional tumor cell kill by chemotherapy & $9.3\times 10^{-6}$ & \citep{dePillisetal2007} \\
     $d_4$      & $\cells^{-1}\,\s^{-1}$ & fractional chemotherapy removal by tumour interaction & $\O(10^{-13})$--$\O(10^{-12})$ & estimated \\
     $m_3$         & $\s^{-1}$     & $\H$ ion removal rate & $\O(10^{-4})$ & \citep{GatenbyGawlinski1996} \\
     $m_4$      & $\s^{-1}$     & chemotherapy removal rate & $\O(10^{-5})$ & \citep{JacksonByrne2000,dePillisetal2006} \\
     $K_1$         & $\cells\,\cm^{-3}$ & normal cell carrying capacity & $5\times 10^7$ & \citep{Tracquietal1995} \\
     $K_2$         & $\cells\,\cm^{-3}$ & tumour cell carrying capacity & $5\times 10^7$ & \citep{Tracquietal1995} \\
     $\alpha_1$         & none          & fractional normal cell death due to tumour cell & $\O(1)$       & chosen freely \\
     $\alpha_2$         & none          & fractional tumour cell death due to normal cell & $\O(1)$       & chosen freely \\
     \bottomrule
     \end{tabularx}
   \label{tab:parameters}%
 \end{table}%

A question arises: What do we choose for $f(N_2)$ and $r_I(s)$? In the model considered by Gatenby~and~Gawlinski~\cite{GatenbyGawlinski1996} and McGillen~et~al.~\cite{McGillenetal2013} it was assumed that acid was produced as a linear function of the tumour cell density, i.e. $f(N_2)=N_2$. In the model considered by Holder~et~al.~\cite{Holderetal2014} a nonlinear acid production term was used as a result of the hypothesis that when the tumour cell density was small, acid was produced at a rate proportional to the tumour cell density until a tumour cell saturation was reached at which point acid production would decrease to zero. With this in mind, the function $f(N_2)=N_2(1-N_2/K_2)$ was used. For simplicity we wish to use the acid production term considered in \citep{GatenbyGawlinski1996} and \citep{McGillenetal2013}, as such we have $f(N_2)=N_2$.

As for $r_I(s)$, we will choose appropriate functions to represent various treatment protocols. Hence the most obvious, and perhaps most realistic, choice would be to chose a function that is periodic, i.e. $r_I(s)=r_I(s+P)$, where $P$ represents the length of the treatment cycle, or period, as this would represent a treatment that occurs in repeated cycles such as taking pills or an intravenous administration made in regularly scheduled doses. However, to enable a greater potential for analysis we can choose $r_I(s)$ to be constant which would represent a constant infusion of chemotherapy drug, i.e. via a device such as an intravenous pump. No matter the choice of $r_I(s)$ we will naturally require it to meet the conditions that $r_I(s)\ge0$ for all $s\ge 0$ and that $r_I(s)$ is bounded almost everywhere. These represent natural limitations on a treatment since a negative infusion rate would represent removal of drug from the system and an unbounded infusion rate would represent an infinite amount of drug to be infused. In the case of administration by pills the use of periodic Dirac delta functions (i.e. $\delta(s)$) can be used to approximate this method of delivery: Let $P$ be the length of the treatment cycle and $N$ being the total number of treatment cycles, then
\begin{equation}\label{eqn:pills}
r_I(s)=r_4\sum_{n=0}^{N-1}{\delta(s-nP)}.
\end{equation}
In the case of intravenous infusion occurring in periodic cycles we can approximate this method of delivery with periodic uses of a boxcar function: Let $P$ denote the cycle period and $s_0$ denote the infusion time, then
\begin{equation}\label{eqn:perinf}
r_I(s)=r_4\sum_{n=0}^{N-1}{\left[\theta(s-nP)-\theta(s-nP-s_0)\right]},
\end{equation}
where $r_4$ represents the constant rate of intravenous infusion and $\theta(s)$ is the Heaviside function.

Considering the function $r_I(s)$ with period $P$ we let
$$
\bar{r}=\frac{1}{P}\int_0^P{r_I(s)\,\d s}
$$
and then utilise the value $\bar{r}$ to non-dimensionalise the equations given by \eqref{eqn:N1}--\eqref{eqn:A}. We remark that this choice of parameter to non-dimensionalise the model enables us to effectively compare the model when utilising different infusion functions. This is because under this non-dimensionalisation the constant infusion rate is equal to the average infusion rate in the periodic case and this will imply that the same amount of drug is infused per cycle no matter the infusion function used. Hence we can compare the models that use the same non-dimensional parameter values. \noeqref{eqn:N2,eqn:H}

Make the following substitutions
\begin{equation}
u_1=\frac{N_1}{K_1},\quad u_2=\frac{N_2}{K_2},\quad u_3=\frac{m_3}{r_3 K_2}H,\quad u_4=\frac{m_4}{\bar{r}}C,\quad t=r_1 s,
\end{equation}
with
\begin{equation}
\beta_2=\frac{r_2}{r_1},\quad \beta_3=\frac{m_3}{r_1},\quad \beta_4=\frac{m_4}{r_1},\quad \delta_1=\frac{r_3K_2d_1}{r_1m_3},\quad \delta_2=\frac{d_2\bar{r}}{r_2m_4},\quad \delta_4=\frac{d_4K_2}{m_4}
\end{equation}
and
\begin{equation}
i(t)=\frac{r_I(t/r_1)}{\bar{r}},\quad \rho=r_1 P.
\end{equation}
We then obtain the following system of non-dimensionalised equations
\begin{equation}\label{eqn:sys}
\bu'=\begin{bmatrix}u_1'\\u_2'\\u_3'\\u_4'
\end{bmatrix}=\begin{bmatrix}
u_1(1-u_1-\alpha_1 u_2-\delta_1 u_3)\\
\beta_2 u_2(1-u_2-\alpha_2 u_1-\delta_2 u_4)\\
\beta_3(u_2-u_3)\\
\beta_4[i(t)-u_4-\delta_4u_4u_2]
\end{bmatrix}=:\bF(t,\bu),
\end{equation}
where $(\,)'$ denotes differentiation with respect to $t$. Note that
\begin{equation}
\bar{i}=\frac{1}{\rho}\int_0^\rho{i(t)\,\d t}=1
\end{equation}
and thus the average rate of infusion over each treatment cycle has been normalised to be equal to one. Moreover, under this non-dimensionalisation, the functions \eqref{eqn:pills}~and~\eqref{eqn:perinf} become
\begin{equation}\label{eqn:pills-non-dim}
i(t)=\rho\sum_{n=0}^{N-1}{\delta(t-n\rho)}
\end{equation}
and
\begin{equation}\label{eqn:perinf-non-dim}
i(t)=\frac{\rho}{\tau}\sum_{n=0}^{N-1}{\left[\theta(t-n\rho)-\theta(t-n\rho-\tau)\right]};\quad \tau=r_1s_0,
\end{equation}
respectively.

A summary of potential non-dimensional parameter values/range of values and interpretation of their meaning has been provided in Table~\ref{tab:nondimparameters}. Note that the primary control parameter is $\delta_2$ since an increase in the amount of drug infused will cause $\delta_2$ to increase.

\begin{table}[htbp]
   \centering
        \caption{Table of non-dimensionalised parameters}
     \begin{tabularx}{\textwidth}{lXl}
     \toprule
     Parameter     & Interpretation & Value/Range    \\
     \midrule
     $\alpha_1$    & fractional normal death due to tumour competition & $\O(1)$        \\
     $\alpha_2$    & fractional tumour death due to normal competition & $\O(1)$        \\
     $\delta_1$    & tumour aggressiveness & $\O(1)$        \\
     $\delta_2$    & chemotherapy aggressiveness & $\O(10^{-1})$--$\O(1)$        \\
     $\delta_4$    & fractional removal due to interaction strength & $\O(10^{-1})$--$\O(1)$         \\
     $\beta_2$     & relative tumour growth rate & $1.0$          \\
     $\beta_3$     & relative $\H$ ion production rate & $\O(10^2)$  \\
     $\beta_4$     & relative chemotherapy rate of increase &  $\O(10)$    \\
     \bottomrule
     \end{tabularx}%
   \label{tab:nondimparameters}%
\end{table}%

Note we define $\R_+=[0,\infty)$ and the convention is used that if $\bu,\bv\in\R^n$, then $\bu\le(<)\bv$ implies that $u_j\le(<)v_j$ for all $j\in\{1,2,3,\ldots,n\}$. Moreover, if $c\in\R$, then $\bu\ge(>) c$ implies that $u_j\ge(>) c$ for all $j\in\{1,2,3,\ldots,n\}$.

\begin{theorem}
Let $i \in C(\R_+,[0,i_M])$, where $i_M\in\R$ and $i_M>0$. If $\bu(0) \in \R_+^4$, then \eqref{eqn:sys} has a unique solution~$\bu$ that satisfies $\bu(t) \in \R_+^4$ for all $t \in \R_+$.
\end{theorem}

\begin{proof}
We utilise Theorem~\ref{thm:Invariance-Exist-Uniqe} that requires the existence of an invariant set, as given by Definition~\ref{def:invariantset}. Clearly, $\bF\in C(\R_+\times S,\R^4)$ and $\bF_\bu'\in C(\R_+\times S,\R^{4^2})$, where $S$ is any compact set in $\R^4$. This implies that $\bF$ is Lipschitz continuous with respect to $\bu$ in any compact set $S\subset\R^4$, that is, there exists a constant $L>0$ such that for any $\bu_1,\bu_2\in S\subset\R^4$ and $t\in\R_+$, the following inequality holds:
\begin{equation}\label{eqn:Lipschitz}
\|\bF(t,\bu_1)-\bF(t,\bu_2)\|\le L\|\bu_1-\bu_2\|.
\end{equation}
The Cauchy--Schwarz inequality and \eqref{eqn:Lipschitz} are now used to show that the one-sided Lipschitz condition in Theorem~\ref{thm:Invariance-Exist-Uniqe} is satisfied on any compact set $S\subset\R^4$. For any $\bu_1,\bu_2\in S\subset\R^4$ and $t\in\R_+$,
\begin{equation}\label{eqn:onesidedLipschitz}
\langle\bu_1-\bu_2,\bF(t,\bu_1)-\bF(t,\bu_2)\rangle\le\|\bu_1-\bu_2\| \|\bF(t,\bu_1)-\bF(t,\bu_2)\|\\
\le L\|\bu_1-\bu_2\|^2.
\end{equation}
An invariant set, as given by Definition~\ref{def:invariantset}, is now constructed in $\R_+^4$. A set $S\subset\R^4$ will be invariant with respect to \eqref{eqn:sys} if
\begin{equation}\label{eqn:tangentcond}
\langle\mathbf{n}(\bu),\bF(t,\bu)\rangle\le0\quad\text{for}\quad t\in\R_+,\quad\bu\in\partial S,
\end{equation}
where $\mathbf{n}(\bu)$ is the outer normal to $S$ at $\bu$. This invariance condition tells us that if $\bu(0)\in S$, then the whole path of the solution $\bu(t)$ will remain in $S$. Let $S=E_1\times E_2 \times E_3 \times E_4$, where $E_1=[0,\max{\{1,u_1(0)\}}]$, $E_2=[0,\max{\{1,u_2(0)\}}]$, $E_3=[0,\max{\{1,u_2(0),u_3(0)\}}]$ and $E_4=[0,\max{\{i_M,u_4(0)\}}]$. Clearly, $S\subset\R_+^4$, $\bu(0)\in S$ and $S$ is compact. Hence $\bF$ will satisfy the one-sided Lipschitz condition on $S$. The boundary of $S$ (i.e $\partial S$) can be written as the union of eight simple sets that have a simple outer normal. Let $\partial S_{i1}=\{\bu\in S: u_i=\inf{E_i}\}$ and $\partial S_{i2}=\{\bu\in S: u_i=\sup{E_i}\}$ for $i=1,2,3,4$, then $\partial S=\bigcup_{i=1}^4{\partial S_{i1}\cup\partial S_{i1}}$. Furthermore, each $\partial S_{ij}$ for $i=1,2,3,4$, $j=1,2$, has outer normal $\mathbf{n}_{ij}=(-1)^j\mathbf{e}_i$, where $\mathbf{e}_i$ for $i=1,2,3,4$ are the standard basis vectors in $\R^4$. Then, it is straightforward to show that
\begin{equation}
\langle\mathbf{n}_{ij},\bF(t,\bu)\rangle=(-1)^jF_i(t,\bu)\le0,~\text{for}~t\in\R_+,~\bu\in\partial S_{ij},~i=1,2,3,4,~j=1,2.
\end{equation}
Hence the set $S$ is invariant and $\bF$ satisfies the one-sided Lipschitz condition on $S$. Therefore by Theorem~\ref{thm:Invariance-Exist-Uniqe}, there exists a unique solution for all time to \eqref{eqn:sys} with initial condition $\bu(0)\in S\subset \R_+^4$, where the path of the solution remains in $S$.
\hfill\end{proof}

Note by a similar argument to the above proof, solutions $\bu(t)$ of \eqref{eqn:sys} are invariant on the sets $\Gamma_1=\{0\}\times\R^3$ and $\Gamma_2=\R\times \{0\}\times \R^2$. Hence if there exists $t_1\in\R_+$ such that $u_1(t_1)=0$, then $u_1(t)=0$ for all $t\in[t_1,\infty)$, similarly if there exists $t_2\in\R_+$ such that $u_2(t_2)=0$, then $u_2(t)=0$ for all $t\in[t_2,\infty)$.

\section{Constant infusion of chemotherapy drug}\label{sec:sys-constant}

If we consider \eqref{eqn:sys} with constant infusion (i.e. $i(t)=1$), we obtain the system of equations
\begin{equation}\label{eqn:sys-const}
\bu'=\begin{bmatrix}u_1'\\u_2'\\u_3'\\u_4'
\end{bmatrix}=\begin{bmatrix}
u_1(1-u_1-\alpha_1 u_2-\delta_1 u_3)\\
\beta_2 u_2(1-u_2-\alpha_2 u_1-\delta_2 u_4)\\
\beta_3(u_2-u_3)\\
\beta_4(1-u_4-\delta_4u_4u_2)
\end{bmatrix}.
\end{equation}

\subsection{Steady-state analysis}

The natural method for analysis of a system of first-order nonlinear autonomous ordinary differential equations~(ODEs) is through the use of a steady-state~(SS) analysis to determine the long-term behaviour of the system. A summary of the results of the SS analysis and stability analysis for system \eqref{eqn:sys-const} is presented below. For full details of the analysis see Lemma~\ref{lem:SSfullsystem} in Appendix~\ref{sec:appendix3}.

System~\eqref{eqn:sys-const} has SS solutions:
\begin{enumerate}
\item[SS1.] $\bu^*=(0,0,0,1)$;
\item[SS2.] $\bu^*=(1,0,0,1)$;
\item[SS3.] $\bu^*=(0,\hat{u}_2,\hat{u}_2,[1+\delta_4 \hat{u}_2]^{-1})$, where $\hat{u}_2$ solves
$$
\delta_4 \hat{u}_2^2+(1-\delta_4)\hat{u}_2+\delta_2-1=0;
$$
\item[SS4.] $\bu^*=(1-(\alpha_1+\delta_1)\tilde{u}_2,\tilde{u}_2,\tilde{u}_2,[1+\delta_4 \tilde{u}_2]^{-1})$, where $\tilde{u}_2$ solves
$$
\delta_4[1-\alpha_2(\alpha_1+\delta_1)]\tilde{u}_2^2+[1-\alpha_2(\alpha_1+\delta_1)+\delta_4(\alpha_2-1)]\tilde{u}_2+\delta_2+\alpha_2-1=0.
$$
\end{enumerate}

However, note that the zero population SS (i.e. SS1) is unconditionally unstable and as a result the solution should always tend towards containing a population of either normal-tissue or tumour-tissue. We see that the tumour free SS (i.e. SS2) is stable provided $\alpha_2+\delta_2>1$. Therefore assuming there is a tumour population, we at the least require this condition to remove the tumour from the system. This condition corresponds to a sufficiently strong treatment in combination with a sufficiently strong population competition provided by the normal-tissue. This state represents the desired state of existence for the system from the point of view of the patient. Hence it is an aim to discover how the system can be altered to make SS2 the most likely long term solution.

It can be seen from Lemma~\ref{lem:SSfullsystem}(iii) that the normal-tissue free SS (i.e. SS3) is stable provided certain parameter conditions are met: these being if
$$
\delta_2<\frac{(\alpha_1+\delta_1+\delta_4)(\alpha_1+\delta_1-1)}{(\alpha_1+\delta_1)^2},
$$
or if
$$
\frac{(\alpha_1+\delta_1+\delta_4)(\alpha_1+\delta_1-1)}{(\alpha_1+\delta_1)^2}\le \delta_2 < \frac{(1+\delta_4)^2}{4\delta_4}\quad\text{and}\quad \frac{1}{\delta_4}<\frac{\alpha_1+\delta_1-2}{\alpha_1+\delta_1}.
$$
This state corresponds to an invasive tumour population in which all normal-tissue in the region is destroyed and replaced by the advancing tumour. Note that from the stability conditions, for the normal-tissue free population to be stable it is a necessary condition that $\alpha_1+\delta_1>1$, otherwise the SS will be unconditionally unstable. This means that the tumour needs to provide sufficiently strong population competition and destructive influence of the acid to potentially be stable. Furthermore, the destructive influence of the treatment ($\delta_2$) needs to be sufficiently small or, alternatively, the removal of chemotherapy drug from interaction with tumour cells ($\delta_4$) needs to be sufficiently large to ensure that the normal-tissue free state is stable. We remark that if $\delta_2>(1+\delta_4)^2/4\delta_4$, then the normal-tissue free state does not exist (i.e. $\hat{u}_2\notin\R$). Hence this condition represents a scenario in which the tumour will be completely removed from the system by the treatment alone. As $\delta_2$ directly relates to the strength of the treatment dose, to obtain a value of $\delta_2$ that will ensure the removal of the tumour by treatment alone may present safety and health concerns for the patient \citep{Perry2008,Devitaetal2010}. However, from the stability conditions, should the tumour-tissue population competition, the destructive influence of the acid or the removal of drug by interaction with the tumour be decreased, then this could enable the tumour to be removed without using a dangerous treatment dose. As noted in \citep{Estrellaetal2013,GatenbyGillies2004} the use of an acid buffer to decrease the acidity could be a potential method to increase the efficacy of the treatment without further increasing doses of strong cytotoxic drugs.

The normal-tissue free state can be stable when the tumour-tissue free state is either unstable or stable. In the case the normal-tissue free state is stable when the tumour-tissue free state is unstable, the long term behaviour would be for the tumour to establish a fixed population that cannot be eradicated by the current treatment protocol. This would suggest that the normal-tissue and chemotherapy treatment would be weak in relation to the tumour-tissue and would potentially correspond to a very aggressive tumour. In the case that both the tumour-tissue free state and the normal-tissue free state are stable, the question of whether the treatment will be effective or the tumour population will successfully invade is dependent on the initial conditions. Therefore the suggestion is that the effectiveness of the treatment will be determined by the size of the initial tumour population. This is consistent with the decreased probability of a cure associated with larger and more established tumour cell populations \cite{Perry2008}.

The coexistence of tumour- and normal-tissue SS (i.e. SS4) is stable and exists for a complicated, yet still calculable, set of parameter conditions given in Lemma~\ref{lem:SSfullsystem}(iv). These parameter conditions suggest that in order for the coexistence state to exist, the system requires unaggressive tumour- and normal-tissue in combination with a weak treatment response. That is, there needs to be very low population competition, tumour aggressiveness and destructive influence of the chemotherapy treatment. As in \citep{GatenbyGawlinski1996}, this suggests that the SS would represent a benign state of existence. The coexistence SS can potentially change to either the tumour-tissue free SS or the normal-tissue free SS provided a sufficient change occurs in the parameters. Should the tumour aggressiveness or the tumour-tissue population competition increase, then the tumour would transition to the invasive state, where the normal-tissue free state is stable. Similarly, should the normal-tissue population competition or the destructive influence of the treatment increase, then the tumour will be eradicated from the system.

\subsection{A reduced model with constant infusion}\label{sec:redsys-constant}

If we consider the situation originally examined in \citep{Byrne2003} we have the system of equations
\begin{equation}\label{eqn:redsys-const}
\bx'=\begin{bmatrix} x_1'\\x_2'\end{bmatrix}=\begin{bmatrix}\beta_2x_1(1-x_1-\delta_2 x_2)\\ \beta_4(1-x_2-\delta_4x_1x_2)\end{bmatrix}.
\end{equation}
If $\alpha_2=0$ in \eqref{eqn:sys-const}, then the reduced system \eqref{eqn:redsys-const} provides the governing dynamics for the tumour-tissue density and cytotoxic drug concentration. Analysing this system we can obtain the conditions under which it is sufficient to obtain tumour clearance from the system by chemotherapy drug without the assistance of population competition. The results for this model are considered in \citep{Byrne2003}, however we provide them here in the current parameters for the convenience of the reader and easy reference for further results in this article. As in \citep{Byrne2003}, see that this system has the SS solutions:
\begin{enumerate}
\item[RS1.] $\bx^*=(0,1)$;
\item[RS2.] $\bx^*=\left(\hat{u}_2,[1+\delta_4 \hat{u}_2]^{-1}\right)$, where $\delta_4 \hat{u}_2^2+(1-\delta_4)\hat{u}_2+\delta_2-1=0$.
\end{enumerate}
Note that these are the same values for the tumour density and drug concentration obtained for SS2 and SS3, respectively. Hence we have that this quadratic equation has the solutions
\begin{equation}
\hat{u}_{2\pm}=\frac{\delta_4-1\pm\sqrt{(1-\delta_4)^2+4\delta_4(1-\delta_2)}}{2\delta_4},
\end{equation}
where $\hat{u}_{2\pm}\in\R$ if and only if $\delta_2 \le (1+\delta_4)^2/4\delta_4$.

As is shown by Byrne~\cite{Byrne2003}, RS1 is stable for $\delta_2>1$, and RS2 is stable if
$$
\delta_2<1,\quad \text{or if} \quad 1\le\delta_2<\frac{(1+\delta_4)^2}{4\delta_4}~\text{and}~\delta_4>1.
$$
Note that the SS $\bx^*=(\hat{u}_{2-},1/(1+\delta_4\hat{u}_{2-}))$ is unconditionally unstable. In the case $1<\delta_2<(1+\delta_4)^2/4\delta_4$~and~$\delta_4>1$ this SS is positive and represents a point on which the separatrix lies.

Furthermore, it can be shown that if $\delta_2\ge1$ and $\delta_4<1$, or if $\delta_2>(1+\delta_4)^2/4\delta_4$, that the long term behaviour of the model will be for the tumour to be eradicated from the system, since not only is RS2 unstable but also biologically meaningless. Moreover, note that under this parameter condition in the case of the full system~\eqref{eqn:sys-const} we similarly get that the tumour-tissue free SS (i.e. SS2) is the only stable solution and as a result the tumour will be eradicated from the system. We remark that in the case of the parameter condition $\delta_2>(1+\delta_4)^2/4\delta_4$, RS2 does not exist as $\hat{u}_2\notin\R$. This would suggest that if this condition is satisfied, then the chemotherapy treatment alone will be sufficient to eradicate the tumour without the assistance of the normal cell population to weaken the tumour cells through competition. We can see that under these conditions that the tumour will always be eradicated since should the population of normal cells be zero, the governing dynamics of the system will reduce to that given by \eqref{eqn:redsys-const}. Therefore this indicates that there is a sufficient scenario under which a tumour will be cleared from the system regardless of the interactions between normal and tumour-tissue. Whilst this observation is an ideal aim to achieve, it is not always feasible or possible due to the fact that this may require doses which would potentially kill the host or require the interaction between the treatment and tumour-tissue to be sufficiently weighted in favour of the treatment.

In the case $\delta_2<1$ we have a situation in which the tumour free solution is unstable and hence this would suggest that the tumour is not able to be removed from the system by chemotherapy. However if we are considering the system given by~\eqref{eqn:sys-const}, provided a condition on population competition is satisfied (i.e. $\alpha_2+\delta_2>1$), the tumour free solution will become stable. Furthermore, should tumour aggressiveness and competition be sufficiently small (see Lemma~\ref{lem:SSfullsystem}(i)) we will have that the normal-tissue free SS will become unstable. Therefore under these conditions we will have that the tumour will be eradicated from the system by the combined strength of population competition and chemotherapy treatment.

\section{Periodic infusion of chemotherapy drug}\label{sec:sys-periodic}

In Section~\ref{sec:ModelForm} we stated that a more realistic function for the infusion of drug is a periodic function such as considered in \citep{dePillisetal2006,dePillisetal2007}. Hence assume that $i(t + \rho) = i(t) \ge 0$ for all $t \in \R_+$. Some preliminary numerical simulations of~\eqref{eqn:sys} were run, with $i(t)$ given by \eqref{eqn:perinf-non-dim}, using the \texttt{ode15s} command in MATLAB with parameter values consistent with Table~\ref{tab:nondimparameters}. The values $\rho=2.8$ ($P$ approximately 1 week), $\tau=1.2$ ($s_0$ approximately 3 days) and total time $T=40$ (approximately 3-4 months) were used with initial values $\bu(0)=(0.9,0.1,0.1,0)$. In these simulations three different behaviours occurred: the eradication of the tumour from the system; the ``invasion'' of the tumour and subsequent destruction of the normal-tissue; the coexistence of the tumour and normal-tissue. Examples of these behaviours are displayed in Figures~\ref{fig:p1}--\ref{fig:p3}, respectively.

\begin{figure}[h]
\centering
{\includegraphics{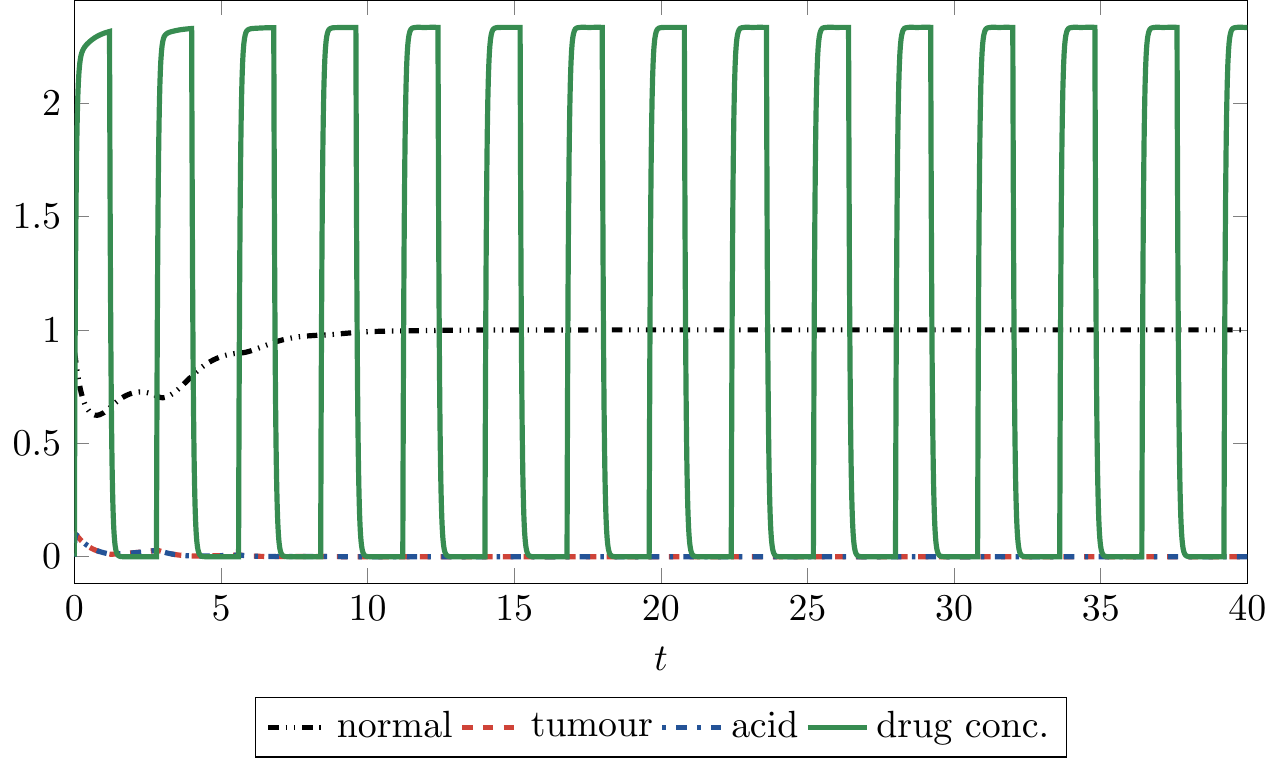}
\caption{Simulation of \protect\eqref{eqn:sys} for $\alpha_1=1$, $\alpha_2=0.5$, $\beta_2=1$, $\beta_3=70$, $\beta_4=20$, $\delta_1=12.5$, $\delta_2=1.1$, $\delta_4=0.6$}
\label{fig:p1}}
\end{figure}

\begin{figure}[h]
\centering
{\includegraphics{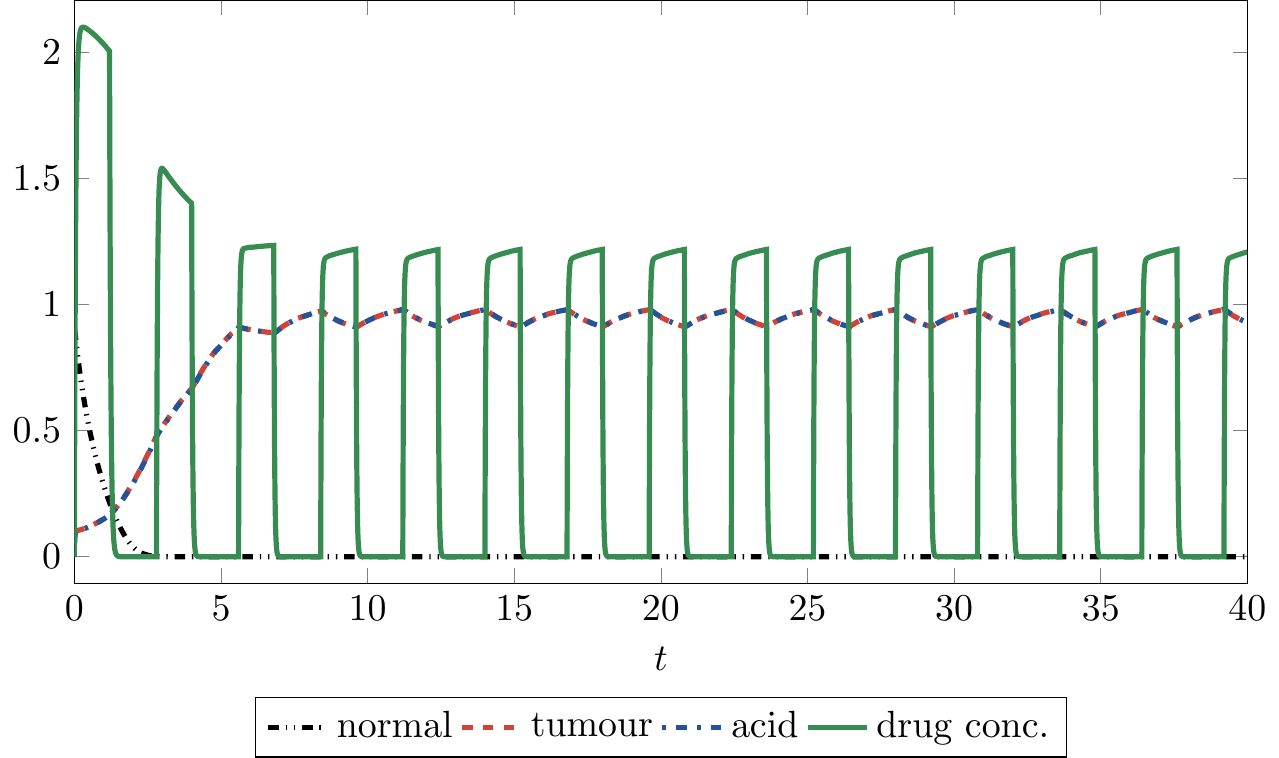}
\caption{Simulation of \protect\eqref{eqn:sys} for $\alpha_1=1$, $\alpha_2=0.5$, $\beta_2=1$, $\beta_3=70$, $\beta_4=20$, $\delta_1=12.5$, $\delta_2=0.1$, $\delta_4=1$}
\label{fig:p2}}
\end{figure}

\begin{figure}[h]
\centering
{\includegraphics{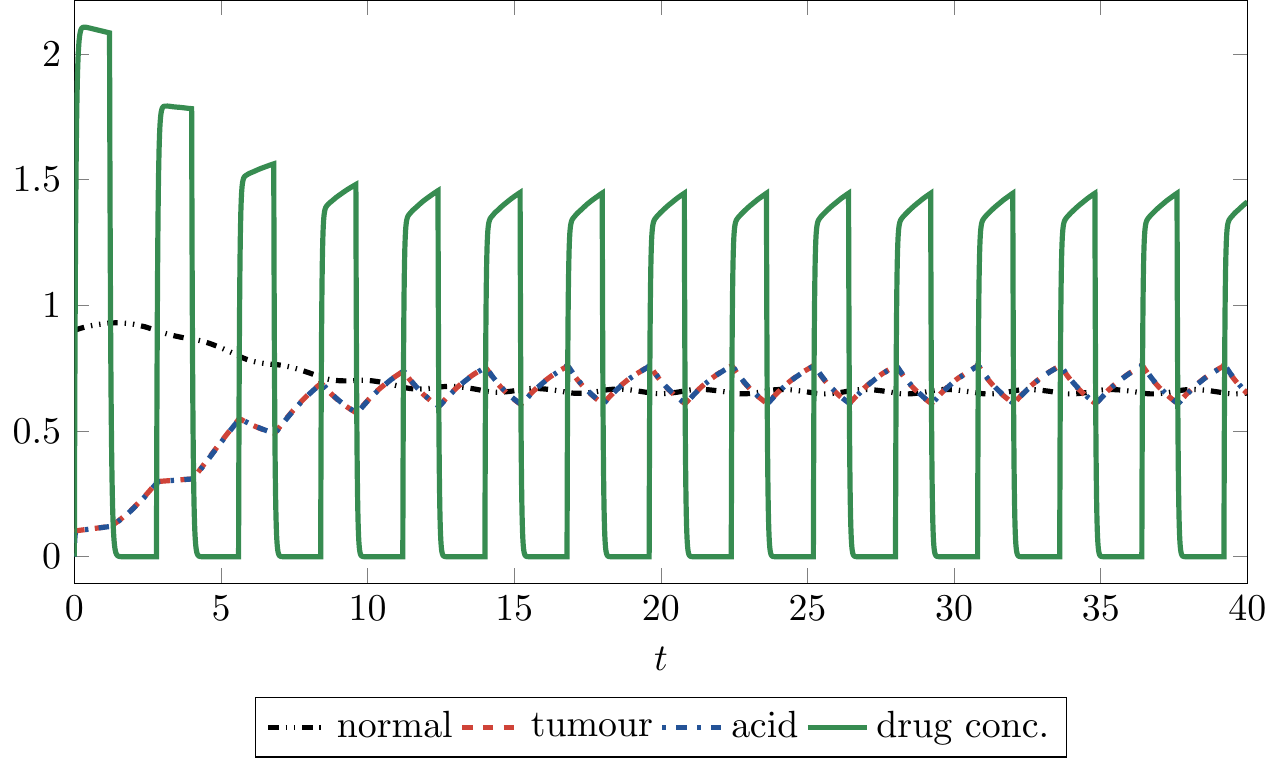}
\caption{Simulation of \protect\eqref{eqn:sys} for $\alpha_1=0.25$, $\alpha_2=0.25$, $\beta_2=1$, $\beta_3=70$, $\beta_4=20$, $\delta_1=0.25$, $\delta_2=0.25$, $\delta_4=1$}
\label{fig:p3}}
\end{figure}

Notice in each of these figures that the solutions evolve towards stable $\rho$-periodic solutions (i.e. $\bu(t)=\bu(t+\rho)$). Therefore to analyse this model we look for time-periodic solutions to \eqref{eqn:sys} with period $\rho$ (i.e. $\bu(t+\rho)=\bu(t)$ for all $t\in \R_+$) and analyse the stability of these solutions to determine the long term behaviour of the system. This is analogous to a steady-state analysis or limit-cycle analysis for an autonomous system of equations.

A reduced version of system \eqref{eqn:sys} is considered first that corresponds to when the solution for $u_1\equiv 0$. The system considered will be analogous to the reduced system considered in Section~\ref{sec:redsys-constant}.  Moreover, the reduced system corresponds to that originally proposed in \citep{Byrne2003}. Byrne~\cite{Byrne2003} however did not analyse the system in this form, but rather made the simplifying assumption that the drug concentration was equivalent to the infusion function which was given by \eqref{eqn:perinf-non-dim}. This reduced the system to a single explicitly solvable Bernoulli equation. Here we present a more thorough analysis of this model for general $\rho$-periodic functions $i\in C(\R_+)$.

\subsection{Existence, uniqueness and stability of the periodic solution of a reduced system}\label{sec:red-sys}

Consider the system
\begin{equation}\label{eqn:redsys}
\bx'=\begin{bmatrix} x_1'\\x_2'\end{bmatrix}=\begin{bmatrix}\beta_2x_1(1-x_1-\delta_2 x_2)\\ \beta_4[i(t)-x_2-\delta_4x_1x_2]\end{bmatrix}=:\bG(t,\bx).
\end{equation}

The results of Lemma~\ref{lem:x1-quadratic} show that periodic solutions can exist for system~\eqref{eqn:redsys} only if $\delta_2<1$, or only if $1\le \delta_2 < (1+\delta_4)^2/4\delta_4$ and $\delta_4>1$. Hence this guides the region of parameter values for which we look for $\rho$-periodic solutions to exist for \eqref{eqn:redsys}.

\subsubsection{Existence}\label{sec:red-sys-exist}

Suppose that $\delta_2 < 1$ and consider the systems
\begin{equation}\label{eqn:redsys2}
\ubx'=\begin{bmatrix} \ubax_1'\\\ubax_2'\end{bmatrix}=\begin{bmatrix}\beta_2\ubax_1(1-\ubax_1-\delta_2 \ubax_2)\\ \beta_4[i(t)-\ubax_2]\end{bmatrix}=:\ubG(t,\ubx)
\end{equation}
and
\begin{equation}\label{eqn:redsys3}
\bbx'=\begin{bmatrix} \bax_1'\\\bax_2'\end{bmatrix}=\begin{bmatrix}\beta_2\bax_1(1-\bax_1)\\ \beta_4[i(t)-\bax_2-\delta_4\bax_1\bax_2]\end{bmatrix}=:\bbG(t,\bbx).
\end{equation}

Consider \eqref{eqn:redsys}, \eqref{eqn:redsys2} and \eqref{eqn:redsys3} with a given initial condition~$\bfeta=(\eta_1,\eta_2)\in\R_+^2$. Following the proof of the existence and uniqueness theorem for the full system~\eqref{eqn:sys}, it can be shown that a unique solution exists for \eqref{eqn:redsys}, \eqref{eqn:redsys2} and \eqref{eqn:redsys3} that are invariant on the region $[0,\max\{1,\eta_1\}] \times [0,\max\{i_M,\eta_2\}]$. Thus if $\bx(0),\ubx(0),\bbx(0)\in[0,1]\times\R_+$, then $\bx(t),\ubx(t),\bbx(t)\in[0,1]\times\R_+$ for all $t\in \R_+$.

Note that the solutions for \eqref{eqn:redsys2} and \eqref{eqn:redsys3} are given by
\begin{equation}\label{eqn:solredsys2}
\ubx(t)=\begin{bmatrix} v(t;\beta_2(1-\delta_2\ubax_2),\beta_2)\\w(t;\beta_4 i,\beta_4) \end{bmatrix}
\end{equation}
and
\begin{equation}\label{eqn:solredsys3}
\bbx(t)=\begin{bmatrix} v(t;\beta_2,\beta_2)\\w(t;\beta_4 i,\beta_4(1+\delta_4\bax_1)) \end{bmatrix},
\end{equation}
where $w$ and $v$ are as in Lemmata~\ref{lem:w}~and~\ref{lem:v}, respectively.

Let $D=[0,1]\times\R_+$ and $M=\diag(1,-1)=M^{-1}$; assume that $\bx(0),\ubx(0),\bbx(0)\in D$ and $M\ubx(0)\le M\bx(0) \le M\bbx(0)$. We claim that
\begin{equation}\label{claim-bounds}
M\ubx(t)\le M\bx(t) \le M\bbx(t)\quad\text{for all}\quad t\in\R_+.
\end{equation}

It is clear that $\bG_\bx',\ubG_\bx',\bbG_\bx'\in C(\R_+\times\R^2,\R^{2^2})$, hence $\bG,\bbG,\ubG$ each satisfy a local Lipschitz condition on any $\Omega\subset\R_+\times\R^4$. It can easily be seen that $M\ubG(t,\bfeta)\le M\bG(t,\bfeta) \le M\bbG(t,\bfeta)$ for all $(t,\bfeta)\in\R_+\times D$. Letting $E=[0,1]\times(-\infty,0]$, we see that the Jacobian matrices $[M\bG(t,M\bfeta)]_\bfeta'$ and $[M\bbG(t,M\bfeta)]_\bfeta'$ are essentially positive (see Remark~\ref{rem:ess-pos}) on $\R_+\times E$, meaning $M\bG(t,M\bfeta)$ and $M\bbG(t,M\bfeta)$ are quasimonotone increasing (see Definition~\ref{def:quasimonotone}) on $\R_+\times E$. Hence by Corollary~\ref{cor:ComparisonCorollary} the claim is proved true.

\begin{figure}[h]
\centering{
\begin{tikzpicture}
		[cube/.style={very thick,black},
			axis/.style={->,blue,thick}]
	\draw[axis] (0,0,0) -- (6,0,0);
	\draw[axis] (0,0,0) -- (0,4.5,0);
	\draw[axis] (0,0,0) -- (0,0,4.5);

	\draw[cube] (1.5,1.5,1.5) -- (1.5,3,1.5) -- (4.5,3,1.5) -- (4.5,1.5,1.5) -- cycle;
	\draw[cube] (1.5,1.5,3) -- (1.5,3,3) -- (4.5,3,3) -- (4.5,1.5,3) -- cycle;

    \node [uoworange] at (4.2,1.8,3) {$S$};

    \draw [red,dashed,thick] plot [smooth, tension=1] coordinates { (3.9,2.4,2.55) (6.1,4.2,3) (4.7,4.35,3.45) (5.1,1.5,2.7) (3.45,0.3,1.65) (2.7,1.35,1.8) (2.55,2.6,2.7)};

    \node [uowdgreen] (initial) at (3.9,2.4,2.55) {\textbullet};

    \node at (1.8,3.2) {$\bfeta =\bu(0;\bfeta)$}
    edge[pil,bend right=45] (initial.west);

    \node [uowpink] (final) at (2.55,2.6,2.7) {\textbullet};

    \node at (-1.5,2.25) {$\varphi(\bfeta) = \bu(\rho;\bfeta)$}
    edge[pil,bend right=45] (final.south west);

    \node (sol) at (3.5,0.1,1.65) {};

    \node (soll) at (4.85,-1.5) {$\bu(t;\bfeta)$};
    \node at (soll.west) {} edge[pil,shorten <=-7pt,bend left=35] (sol.south);

	\draw[cube] (1.5,1.5,1.5) -- (1.5,1.5,3);
	\draw[cube] (1.5,3,1.5) -- (1.5,3,3);
	\draw[cube] (4.5,1.5,1.5) -- (4.5,1.5,3);
	\draw[cube] (4.5,3,1.5) -- (4.5,3,3);
	
\end{tikzpicture}
\caption{Diagram showing solutions $\bu$ initially in $S\subset\R^3$ will be in $S$ at time $\rho$}
\label{fig:fixed-point}}
\end{figure}
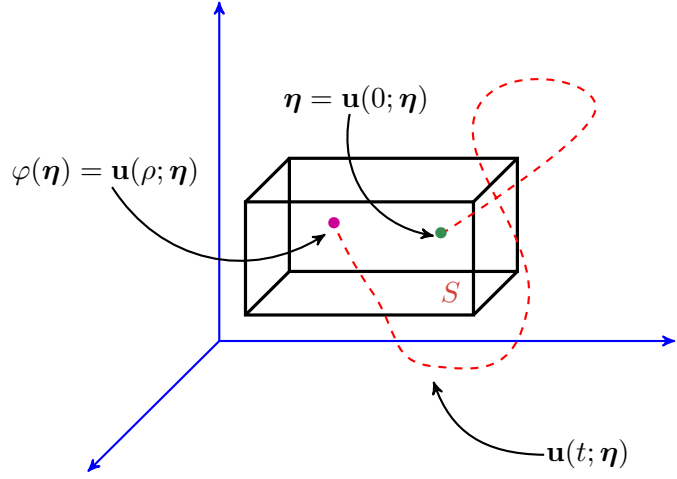

After having established the time-dependent bounds~\eqref{claim-bounds} on the solution of \eqref{eqn:redsys}, we are ready to prove the actual existence of a $\rho$-periodic solution to \eqref{eqn:redsys}. Define the rectangle
$$
R = \{\bfeta \in \R_+^2 : M\ubx(0) \le M\bfeta \le M\bbx(0)\}.
$$
Let $\bfeta = \bx(0) \in R$. Denote by $\bx(\cdot;\bfeta)$ the solution of \eqref{eqn:redsys} with initial condition~$\bfeta \in R$. Define a map~$\varphi : R \rightarrow \R_+^2$ by
$$
\varphi(\bfeta) := \bx(\rho;\bfeta)
$$
for every $\bfeta \in R$. We wish to apply Brouwer's Fixed Point Theorem to ensure the existence of a fixed point of $\varphi$, i.e. we want to show that there is some initial condition~$\bfeta_0 \in R$ for which
$$
\varphi(\bfeta_0) = \bx(\rho;\bfeta_0) = \bfeta_0.
$$
For this particular initial condition, $\bx(\rho)=\bx(0)$. This idea is illustrated in Figure~\ref{fig:fixed-point} for $\bu\in\R^3$. Then a result from \citep{Farkas1994} will enable us to conclude that $\bx(t)=\bx(t+\rho)$ for all $ t\in\R_+$.

It is clear that $\varphi$ is continuous on $R$. Using Lemmata~\ref{lem:w}~and~\ref{lem:v}, if we pick
$$
\ubx(0)=\begin{bmatrix} v(0;\beta_2(1-\delta_2\ubax_2),\beta_2) \\w(0;\beta_4 i,\beta_4) \end{bmatrix}=\begin{bmatrix} \frac{1 - \re^{-\int_0^\rho g(s') \, \d s'}}{\beta_2 \int_0^\rho \re^{-\int_s^\rho g(s') \, \d s'} \, \d s} \\  \frac{\beta_4}{\re^{\beta_4 \rho} - 1} \int_0^\rho{\re^{\beta_4 s} i(s) \, \d s}\end{bmatrix},
$$
where $g(s)=\beta_2[1-\delta_2 \ubax_2(s)]$, then $\ubx(t+\rho)=\ubx(t)>0$ for all $t\ge 0$. This is true provided $\int_0^\rho g(s) \, \d s > 0$. From Lemma~\ref{lem:w}, we see that
$$
\int_0^\rho \ubax_2(s) \, \d s = \int_0^\rho i(s) \, \d s = \rho.
$$
Recalling that $\delta_2 < 1$, we obtain
$$
\delta_2 \int_0^\rho \ubax_2(s) \, \d s < \rho, \quad \text{or} \quad \int_0^\rho g(s) \, \d s = \int_0^\rho \beta_2 [1 - \delta_2 \ubax_2(s)] \, \d s > 0.
$$

Similarly, from Lemmata~\ref{lem:w}~and~\ref{lem:v}, if we choose
$$
\bbx(0)=\begin{bmatrix}v(0;\beta_2,\beta_2)\\w(0;\beta_4 i,\beta_4 (1+\delta_4\bax_1))  \end{bmatrix}=\begin{bmatrix} 1 \\  \frac{\beta_4}{\re^{(1 + \delta_4) \beta_4 \rho} - 1} \int_0^\rho{\re^{(1 + \delta_4) \beta_4 s} i(s) \, \d s}\end{bmatrix},
$$
then $\bbx(t+\rho)=\bbx(t)>0$ for all $t\ge 0$. Furthermore, as $\beta_4<\beta_4[1+\delta_4\bax_1(t)]$  and $\beta_2[1-\delta_2\ubax_2(t)]<\beta_2$ for $t\in\R_+$ we have from Lemmata~\ref{lem:w}~and~\ref{lem:v} that $M\ubx(t)\le M\bbx(t)$ for all $t\in \R_+$. With the above choices for $\ubx(0)$ and $\bbx(0)$ we have that $R\subset \inte(\R_+^2)$ and
$$
M\ubx(0)=M\ubx(\rho)\le M\bx(\rho;\bfeta)\le M\bbx(\rho)=M\bbx(0),
$$
that is, $\varphi(\bfeta)\in R$, which implies that $\varphi(R)\subset R$. By Brouwer's Fixed Point Theorem there is some initial condition $\bfeta_0\in R\subset \inte(\R_+^2)$ for which
$$
\varphi(\bfeta_0)=\bx(\rho;\bfeta_0)=\bfeta_0.
$$
For this particular initial condition, $\bx(\rho)=\bx(0)$. Then from \citep[Lemma~2.2.1]{Farkas1994} we have
$$
\bx(t+\rho)=\bx(t)>0\quad\text{for all} \quad t\in \R_+,
$$
thus showing the existence of a strictly positive $\rho$-periodic solution to \eqref{eqn:redsys}.

Now we prove the uniqueness of the solution constructed above.

\subsubsection{Uniqueness}

Note that if $\bu(t)\ge 0$ and $\bv(t)\ge 0$ are solutions to \eqref{eqn:redsys} and $\bu(t_0)=\bv(t_0)$ at some $t_0\in\R_+$, then $\bu(t)=\bv(t)$ for all $t\in\R_+$ by uniqueness.

Now, let $\bu(t),\bv(t)>0$ be $\rho$-periodic solutions of \eqref{eqn:redsys}. It will be shown that $u_1(t)=v_1(t)$ for all $t\in\R_+$ if and only if $u_2(t)=v_2(t)$ for all $t\in\R_+$. If $u_1(t)=v_1(t)$ for all $t\in\R_+$, then from \eqref{eqn:redsys} and periodicity it can be seen that
\begin{equation}
0=\int_{0}^\rho{[u_{1}(s)-v_{1}(s)]\,\d s}+\delta_2\int_{0}^\rho{[u_{2}(s)-v_{2}(s)]\,\d s}=\delta_2\int_{0}^\rho{[u_{2}(s)-v_{2}(s)]\,\d s}.
\end{equation}
Hence by the Mean Value Theorem there exists $t_0\in(0,\rho)$ such that $u_{2}(t_0)-v_{2}(t_0)=0$ (i.e. $u_{2}(t_0)=v_{2}(t_0)$) and since $u_1(t)=v_1(t)$ for all $t\in\R_+$ it follows that $\bu(t_0)=\bv(t_0)$ which implies $\bu(t)=\bv(t)$ for all $t\in\R_+$, i.e. $u_2(t)=v_2(t)$ for all $t\in\R_+$. It can be shown similarly that if $u_2(t)=v_2(t)$ for all $t\in\R_+$, then $u_1(t)=v_1(t)$ for all $t\in\R_+$.

Let $\bu(t)\ge 0$ and $\bv(t)\ge 0$ be distinct $\rho$-periodic solutions of \eqref{eqn:redsys}. It will now be shown that $u_1(t)\le(\ge) v_1(t)$ for all $t\in\R_+$ and there exists $t_1\in[0,\rho)$ such that $u_1(t_1)<(>) v_1(t_1)$. Since $\bu(t)$ and $\bv(t)$ are $\rho$-periodic it is sufficient to show $u_1(t)\le(\ge) v_1(t)$ for all $t\in[0,\rho)$.

Assume that there exists $t_2\in[0,\rho)$ such that $u_1(t_2)=v_1(t_2)$, then as $\bu(t)$ and $\bv(t)$ are distinct we must have $u_2(t_2)\ne v_2(t_2)$. Assume $u_2(t_2)>(<)u_2(t_2)$, then letting $M=\diag(1,-1)$ we have $M\bv(t_2)\ge(\le)M\bu(t_2)$, which shows by Theorem~\ref{thm:ComparisonTheorem} that $M\bv(t)\ge(\le)M\bu(t)$ for all $t\in[t_2,\infty)$ (i.e. $u_1(t)\le(\ge) v_1(t)$ and $u_2(t)\ge(\le) v_2(t)$ for all $t\in[t_2,\infty)$), and by periodicity of $\bu(t)$ and $\bv(t)$ this must hold for all $t\in\R_+$. Furthermore since $u_1(t)=v_1(t)$ for $t\in\R_+$ implies $\bu(t)$ and $\bv(t)$ are not distinct there must exist $t_1\in[0,\rho)$ such that $u_1(t_1)<(>)v_1(t_1)$. Now if $u_1(t)\ne v_1(t)$ for any $t\in[0,\rho)$, then as a consequence of the continuity of $\bu(t)$ and $\bv(t)$ and the Intermediate Value Theorem $u_1(t)<(>)v_1(t)$ for all $t\in\R_+$.

Assume that $\bu(t)>0$ and $\bv(t)>0$ are distinct $\rho$-periodic solutions to \eqref{eqn:redsys}, then as shown previously this implies without loss of generality that $u_1(t)\le v_1(t)$ for all $t\in\R_+$ and there exists $t_1\in[0,\rho)$ such that $u_1(t_1)< v_1(t_1)$. From Lemma~\ref{lem:x1-quadratic}, $u_{1}(t)$ is given by the implicit form
\begin{equation}\label{eqn:u1equation}
u_{1}(t)=\frac{\delta_4-1+\sqrt{(1+\delta_4)^2+4\delta_4[f(t)-\delta_2i(t)]}}{2\delta_4},\quad f(t)=\Der{}{t}V(\bu(t))
\end{equation}
and $v_{1}(t)$ is given by the implicit form
\begin{equation}\label{eqn:v1equation}
v_{1}(t)=\frac{\delta_4-1+\sqrt{(1+\delta_4)^2+4\delta_4[g(t)-\delta_2i(t)]}}{2\delta_4},\quad g(t)=\Der{}{t}V(\bv(t)),
\end{equation}
noting that $f$ and $g$ must be continuous by the continuity of $\bu>0$ and $\bv>0$. Since $u_1(t)\le v_1(t)$ for all $t\in\R_+$ and there exists $t_1\in[0,\rho)$ such that $u_1(t_1)< v_1(t_1)$, then \eqref{eqn:u1equation} and \eqref{eqn:v1equation} implies $f(t)\le g(t)$ for all $t\in\R_+$ and $f(t_1)< g(t_1)$. By continuity this implies
\begin{equation}\label{eqn:intoffg}
\int_{0}^\rho{[g(s)-f(s)]\,\d s}>0.
\end{equation}
However by the periodicity of $\bu(t)$ and $\bv(t)$
\begin{equation}\label{eqn:intoffg-eq}
\int_{0}^\rho{[g(s)-f(s)]\,\d s}=0,
\end{equation}
which is a contradiction, hence $u_1(t)$ and $v_1(t)$ cannot be distinct (i.e. $u_1(t)=v_1(t)$) which implies that $u_2(t)$ and $v_2(t)$ are not distinct (i.e. $u_2(t)=v_2(t)$).

We have therefore proved the following theorem:
\begin{theorem}\label{thm:red-sys-exist-unique}
Suppose that $0 < \delta_2 < 1$. Then \eqref{eqn:redsys} has a unique solution~$\bx$ that satisfies
$$
\bx(t + \rho) = \bx(t) > 0\quad \text{for all} \quad t\in\R_+.
$$
\end{theorem}

\subsubsection{Stability}\label{sec:red-sys-stab}

Here we prove the stability of the strictly positive $\rho$-periodic solution of \eqref{eqn:redsys} by utilising \citep[Theorem 4.2.1]{Farkas1994}. We summarise the required results of \cite{Farkas1994} below.

Let
\begin{equation}\label{eqn:systemperiodic}
\bu'=\bF(t,\bu),
\end{equation}
where $\bF\in C(\R\times X,\R^n)$, $\bF_\bu'\in C(\R\times X,\R^{n^2})$, $X$ is an open connected subset of $\R^n$ and $\bF(t,\cdot)=\bF(t+\rho,\cdot)$. Let $\bu:\R\to X$ be a non-constant $\rho$-periodic solution to \eqref{eqn:systemperiodic}. Then making the coordinate transformation $\bz=\bu-\mathbf{p}(t)$ we have
$$
\bz'=\bF(t,\bz+\mathbf{p}(t))-\bF(t,\mathbf{p}(t))=\bF_\bu'(t,\mathbf{p}(t))\bz+o(|\bz|).
$$
Hence the linearisation of \eqref{eqn:systemperiodic} at $\mathbf{p}(t)$ is given by
\begin{equation}\label{eqn:linearisationperiodic}
\by'=\bF_\bu'(t,\mathbf{p}(t))\by.
\end{equation}
If $\Phi(t)$ represents the fundamental matrix solution of \eqref{eqn:linearisationperiodic}, then the \textit{characteristic multipliers} of \eqref{eqn:linearisationperiodic} are given by the eigenvalues of $\Phi(\rho)$.

From \citep[Theorem 4.2.1]{Farkas1994}, if all the characteristic multipliers of system \eqref{eqn:linearisationperiodic} are in modulus less than 1 (i.e. the spectral radius of $\Phi(\rho)$ is less than 1), then $\mathbf{p}$ is a uniformly asymptotically stable solution of \eqref{eqn:systemperiodic}; if \eqref{eqn:linearisationperiodic} has at least one characteristic multiplier with modulus greater than 1 (i.e. the spectral radius of $\Phi(\rho)$ is greater than 1), then $\mathbf{p}$ is unstable.

Consider \eqref{eqn:redsys}, which when linearised about a strictly positive $\rho$-periodic solution $\bx$ produces the system
\begin{equation}\label{eqn:red-sys-lin}
\frac{\d \bz}{\d t}= \begin{bmatrix}
\beta_2 [1 - 2 x_1(t) - \delta_2 x_2(t)] & -\beta_2 \delta_2 x_1(t)\\
-\beta_4 \delta_4 x_2(t) & -\beta_4 [1 + \delta_4 x_1(t)]
\end{bmatrix}\bz =: A(t) \bz.
\end{equation}

The fundamental matrix~$\Phi(t)$ of this system satisfies
$$
\Der{\Phi}{t} = A(t) \Phi, \quad \Phi(0) = I.
$$

Let $P$ be a $2 \times 2$ invertible, differentiable matrix function such that $P(t + \rho) = P(t)$ for all $t\in\R_+$. If we let $\by(t) = P(t) \bz(t)$, then $\by$ satisfies
\begin{equation}\label{eqn:red-sys-lin-trans}
\Der{\by}{t} = [P'(t) + P(t) A(t)] P^{-1}(t)\by=:B(t) \by.
\end{equation}
The fundamental matrix~$\Psi(t)$ of this system satisfies
$$
\Der{\Psi}{t} = B(t) \Psi, \quad \Psi(0) = I,
$$
where $I$ is the identity matrix. It can easily be shown that $P(t)\Phi(t)P^{-1}(0)$ is a fundamental matrix solution of \eqref{eqn:red-sys-lin-trans}, that is, $\Psi(t)=P(t)\Phi(t)P^{-1}(0)$. Since $P(t+\rho)=P(t)$, it is clear that $\Psi(\rho)=P(\rho)\Phi(\rho)P^{-1}(\rho)$ (i.e. $\Psi(\rho)$ is similar to $\Phi(\rho)$), hence $\Psi(\rho)$ and $\Phi(\rho)$ have the same eigenvalues. This demonstrates the requirement for $P$ to be $\rho$-periodic.

We let
$$
P(t) =
\begin{bmatrix}
p_{11}(t) & 0\\
0 & p_{22}(t)
\end{bmatrix}
$$
for some $\rho$-periodic functions $p_{11}$ and $p_{22}$ to be determined, which implies
$$
B =
\begin{bmatrix}
\frac{p_{11}'}{p_{11}} + a_{11} & \frac{p_{11} a_{12}}{p_{22}}\\
\frac{p_{22} a_{21}}{p_{11}} & \frac{p_{22}'}{p_{22}} + a_{22}
\end{bmatrix}.
$$
We want $p_{11}(t) p_{22}(t) < 0$ so that $b_{12}(t) > 0$ and $b_{21}(t) > 0$ (i.e. $B$ is essentially positive). Since $B(t)$ is essentially positive for all $t\in\R_+$, the same argument as that used in \citep[p. 190]{Farkas1994} shows that each entry of $\Psi(t)$ is positive for $t \in [0,\rho]$. In particular, each entry of $\Psi(\rho)$ is positive. Let $\lambda_1, \lambda_2$ denote the eigenvalues of $\Psi(\rho)$, that is, the characteristic multipliers of~\eqref{eqn:red-sys-lin-trans}.

By Perron's Theorem, $\Psi(\rho)$ has a unique largest positive eigenvalue~$\lambda_2$, say, with a corresponding eigenvector~$\bv = \begin{bmatrix} v_1 & v_2 \end{bmatrix}^{\T}$ having strictly positive components such that $| \lambda_1 | < \lambda_2$. Hence for $\bx$ to be stable, we need to show that $\lambda_2 < 1$.

Let $\by(t) = \Psi(t) \bv$, where $\bv$ is such that $\Psi(\rho)\bv = \lambda_2 \bv$. Then $\by$ satisfies $\d \by/\d t = B(t) \by$. If $\by(t) = \begin{bmatrix} y_1(t) & y_2(t) \end{bmatrix}^{\T}$, it follows that $\by(t)> 0$ for all $t \in[0,\rho]$. Suppose that for the moment that we can find a $\rho$-periodic function~$\bfxi(t)=\begin{bmatrix} \xi_1(t) & \xi_2(t) \end{bmatrix}^{\T}$ such that
$$
\inner{\bfxi(0)}{\by(0)} > 0\quad\text{and}\quad\int_0^\rho \frac{\d }{\d t} \inner{\bfxi(t)}{\by(t)} \, \d t < 0.
$$
Then $\inner{\bfxi(\rho)}{\by(\rho)} < \inner{\bfxi(0)}{\by(0)}$; however $\by(\rho) = \Psi(\rho) \bv = \lambda_2 \bv$, so that $\inner{\bfxi(0)}{\lambda_2\bv} < \inner{\bfxi(0)}{\bv}$ and this would show that $\lambda_2 < 1$. We now proceed to find the desired function~$\bfxi(t)$. We have
\begin{align*}
\inner{\bfxi}{\by}' = &\, \inner{\bfxi'}{\by} + \inner{\bfxi}{\by'}=\inner{\bfxi'}{\by} + \inner{\bfxi}{B\by}\\
= &\, \inner{\bfxi'}{\by} + \inner{B^{\T}\bfxi}{\by}=\inner{\bfxi'+B^{\T}\bfxi}{\by}\\
= &\, \left[\xi_1' + \left(\frac{p_{11}'}{p_{11}} + a_{11}\right) \xi_1 + \frac{p_{22} a_{21}}{p_{11}} \xi_2\right] y_1\\
&\,+ \left[\xi_2' + \frac{p_{11} a_{12}}{p_{22}} \xi_1 + \left(\frac{p_{22}'}{p_{22}} + a_{22}\right) \xi_2\right] y_2.
\end{align*}

Take
$$
p_{11}(t) = \frac{1}{x_1(t)}, \quad p_{22}(t) = -\beta_2 \delta_2, \quad \xi_1(t) = \beta_4 [1 + \delta_4 x_1(t)], \quad \xi_2(t) = 1
$$
for example, which are all $\rho$-periodic and $p_{11}(t) p_{22}(t) < 0$ for all $t \in\R_+$. We also note that $\inner{\bfxi(0)}{\by(0)}>0$. This makes the coefficient of $y_2$ equal to zero, so that
\begin{align*}
\inner{\bfxi}{\by}' & = \beta_2 \beta_4 x_1 (\delta_4 - 1 - 2 \delta_4 x_1)y_1
\end{align*}
after some algebra.

Suppose that $\delta_2<1$. Then from Lemma~\ref{lem:x1-quadratic}
$$
x_1(t)=\frac{(\delta_4 - 1) + \sqrt{(\delta_4 - 1)^2 + 4 \delta_4 [1-\delta_2 i(t) + h(t)]}}{2 \delta_4},\quad\text{where}\quad h(t)=\Der{}{t}V(\bx(t)),
$$
which implies that
$$
\delta_4 - 1 - 2 \delta_4 x_1(t)=-\sqrt{(1-\delta_4)^2+4\delta_4[1-\delta_2 i(t)+h(t)]}\le 0.
$$
Moreover, from the periodicity of $\bx$ it can be seen that $(1/\rho)\int_0^\rho{[\delta_2 i(s)-h(s)]\,\d s}=\delta_2$. Thus from the Mean Value Theorem there exists $t^*\in(0,\rho)$ such that
$$
(1-\delta_4)^2+4\delta_4[1-\delta_2 i(t^*)+h(t^*)]=(1-\delta_4)^2+4\delta_4(1-\delta_2).
$$
Since $\delta_2<1$, it follows that $(1-\delta_4)^2+4\delta_4(1-\delta_2)>0$ and then by continuity, there exists $\epsilon > 0$ such that
$$
(1-\delta_4)^2+4\delta_4[1-\delta_2 i(t)+h(t)]>0 \quad \text{for all}\quad t \in (t^* - \epsilon,t^* + \epsilon).
$$
This yields, by the strict positivity of $\bx$ and $\by$,
$$
\int_{t^* - \epsilon}^{t^* + \epsilon} x_1(t) [\delta_4 - 1 - 2 \delta_4 x_1(t)] y_1(t) \, \d t < 0.
$$
Thus,
\begin{align*}
\int_0^\rho \frac{\d }{\d t} \inner{\bfxi(t)}{\by(t)} \, \d t & = \int_0^\rho \beta_2 \beta_4 x_1(t) [\delta_4 - 1 - 2 \delta_4 x_1(t)]y_1(t) \, \d t\\
& = \int_0^{t^* - \epsilon}\beta_2 \beta_4 x_1(t) [\delta_4 - 1 - 2 \delta_4 x_1(t)] y_1(t) \, \d t\\
& \quad {} + \int_{t^* - \epsilon}^{t^* + \epsilon} \beta_2 \beta_4 x_1(t) [\delta_4 - 1 - 2 \delta_4 x_1(t)]y_1(t) \, \d t\\
& \quad {} + \int_{t^* + \epsilon}^{\rho} \beta_2 \beta_4 x_1(t) [\delta_4 - 1 - 2 \delta_4 x_1(t)]y_1(t) \, \d t\\
& < 0.
\end{align*}
The above argument then shows that if $\delta_2<1$ and a strictly positive $\rho$-periodic solution exists to \eqref{eqn:redsys}, then $|\lambda_1| < \lambda_2 < 1$ (i.e. $\bx$ is asymptotically stable). Therefore the unique $\rho$-periodic solution of \eqref{eqn:redsys} found in Section~\ref{sec:red-sys-exist} is asymptotically stable.

\begin{remark}
As a result of Lemma~\ref{lem:x1-quadratic} any strictly positive $\rho$-periodic solution $\bx$ to \eqref{eqn:redsys} must satisfy the implicit form
\begin{equation}\label{eqn:x1-sol-pm}
x_1(t)=\frac{(\delta_4 - 1) \pm \sqrt{(\delta_4 + 1)^2 - 4 \delta_4 [\delta_2 i(t) - h(t)]}}{2 \delta_4},\quad\text{where}\quad h(t)=\Der{}{t}V(\bx(t)).
\end{equation}
By a similar argument to the above it can be concluded that if $\delta_2<1$ or if $1\le \delta_2 < (1+\delta_4)^2/4\delta_4$ and $\delta_4>1$ and a strictly positive $\rho$-periodic solution exists for \eqref{eqn:redsys} that satisfies the ``plus'' case of \eqref{eqn:x1-sol-pm}, then $|\lambda_1| < \lambda_2 < 1$ (i.e. $\bx$ is asymptotically stable). Similarly, it can be concluded that if a strictly positive $\rho$-periodic solution exists for \eqref{eqn:redsys} that satisfies the ``minus'' case of \eqref{eqn:x1-sol-pm}, then $\lambda_2 > 1$ (i.e. $\bx$ is unstable).
\end{remark}

\subsection{Existence of co-existence periodic solution}\label{sec:coexit-exist}

The existence of a strictly positive $\rho$-periodic solution to the full system~\eqref{eqn:sys} will be shown in this section.

Note that \eqref{eqn:sys} is invariant on $\R_+^4$. Consider the systems
\begin{equation}\label{eqn:sys2}
\ubu'=\begin{bmatrix}\ubau_1'\\\ubau_2'\\\ubau_3'\\\ubau_4'
\end{bmatrix}=\begin{bmatrix}
\ubau_1(1-\ubau_1-\alpha_1 \ubau_2-\delta_1 \ubau_3)\\
\beta_2 \ubau_2(1-\ubau_2-\delta_2 \ubau_4)\\
\beta_3(\ubau_2-\ubau_3)\\
\beta_4[i(t)-\ubau_4-\delta_4\ubau_4\ubau_2]
\end{bmatrix}=:\ubF(t,\ubu),
\end{equation}
and
\begin{equation}\label{eqn:sys3}
\bbu'=\begin{bmatrix}\bau_1'\\\bau_2'\\\bau_3'\\\bau_4'
\end{bmatrix}=\begin{bmatrix}
\bau_1(1-\bau_1-\alpha_1 \bau_2)\\
\beta_2 \bau_2(1-\bau_2-\alpha_2 \bau_1-\delta_2 \bau_4)\\
\beta_3(\bau_2-\bau_3)\\
\beta_4[i(t)-\bau_4]
\end{bmatrix}=:\bbF(t,\bbu).
\end{equation}
Similarly to \eqref{eqn:sys} the solutions to \eqref{eqn:sys2} and \eqref{eqn:sys3} can be shown to be unique and invariant on $\R_+^4$.

We now wish to establish the existence of strictly positive $\rho$-periodic solutions to \eqref{eqn:sys2} and \eqref{eqn:sys3}. First consider \eqref{eqn:sys2} and let $\ubu(t)$ denote a solution, where $\ubu(0)\in \R_+^4$. From the analysis of the normal-tissue free system, if $\delta_2<1$ there is a unique strictly positive $\rho$-periodic solution to $\ubau_2'=\ubaf_2(t,\ubu)$ and $\ubau_4'=\ubaf_4(t,\ubu)$. Then from Lemma~\ref{lem:w} there exists a unique $\rho$-periodic solution to $\ubau_3'=\ubaf_3(t,\ubu)$, where $\ubau_3(t)=w(t;\beta_3 \ubau_2,\beta_3)$ with initial condition given by \eqref{eqn:w-ic}. Using these $\rho$-periodic solutions and Lemma~\ref{lem:v} if $\int_0^\rho{[1-\alpha_1 \ubau_2(s)-\delta_1 \ubau_3(s)]\,\d s}>0$, there exists a unique strictly positive $\rho$-periodic solution to $\ubau_1'=\ubaf_1(t,\ubu)$, where $\ubau_1(t)=v(t; 1-\alpha_1\ubau_2-\delta_1\ubau_3,1)$ with initial condition given by \eqref{eqn:v-ic}. Consider
\begin{equation}\label{eqn:alpha1delta1cond}
\int_0^\rho{[1-\alpha_1 \ubau_2(s)-\delta_1 \ubau_3(s)]\,\d s}=\int_0^\rho{[1-(\alpha_1+\delta_1) \ubau_2(s)]\,\d s}\ge 1-(\alpha_1+\delta_1)\hat{u}_{2+}
\end{equation}
from Lemmata~\ref{lem:w}~and~\ref{lem:x1-quadratic}. Hence if $\delta_2<1$ and $1-(\alpha_1+\delta_1)\hat{u}_{2+}>0$, then there exists a unique strictly positive $\rho$-periodic solution to \eqref{eqn:sys2}.

Now consider system \eqref{eqn:sys3} and let $\bbu(t)$ denote a solution, where $\bbu(0)\in \R_+^4$. From the analysis of the normal-tissue free problem it is known that there exists a unique strictly positive $\rho$-periodic solution to $\bau_4'=\baf_4(t,\bbu)$. Then from \citep[Prop. 36.1 and 36.3]{Hess1991} there exists a unique strictly positive $\rho$-periodic solution to $\bau_1'=\baf_1(t,\bbu)$ and $\bau_2'=\baf_2(t,\bbu)$ if
\begin{equation}\label{eqn:condCE1}
1<(>)\alpha_1\frac{1}{\rho}\int_0^\rho{[1-\delta_2 \bau_4(s)]\,\d s}=\alpha_1(1-\delta_2)
\end{equation}
and
\begin{equation}\label{eqn:condCE2}
\frac{1}{\rho}\int_0^\rho{[1-\delta_2 \bau_4(s)]\,\d s}=1-\delta_2<(>)\alpha_2.
\end{equation}
Noting $\bau_2(t)>0$ for all $t\in\R_+$, we have from Lemma~\ref{lem:w} that $\bau_3'=\baf_3(t,\bbu)$ has a unique strictly positive $\rho$-periodic solution given by $\bau_3(t)=w(t;\beta_3 \bau_2,\beta_3)$ with initial condition given by \eqref{eqn:w-ic}.

Assume that the parameter conditions
\begin{equation}\label{eqn:sys-coexist-par-con}
\delta_2<1,\quad 1-(\alpha_1+\delta_1)\hat{u}_{2+}>0, \quad1<(>)\alpha_1(1-\delta_2)\quad \text{and} \quad1-\delta_2<(>)\alpha_2
\end{equation}
are satisfied, then there exists unique strictly positive $\rho$-periodic solutions to \eqref{eqn:sys2} and \eqref{eqn:sys3} denoted by $\ubu(t)$ and $\bbu(t)$, respectively. Letting $M=\diag(1,-1,-1,1)=M^{-1}$, we wish to show that $M\ubu(t)\le M\bbu(t)$ for all $t\in\R_+$. It was shown in the analysis of the normal-tissue free $\rho$-periodic solution that $\ubau_4(t)\le \bau_4(t)$ for all $t\in\R_+$. Consider the $\rho$-periodic solutions $\ubau_2(t)$ and $\bau_2(t)$ that satisfy
$$
\ubau_2'=\beta_2\ubau_2[1-\ubau_4(t)-\ubau_2]\quad \text{and}\quad \bau_2'=\beta_2\bau_2[1-\alpha_2\bau_1(t)-\delta_2\bau_4(t)-\bau_2].
$$
Note that $1-\delta_2 \bau_4(t)-\alpha_2 \bau_1(t)<1-\delta_2 \ubau_4(t)$ for all $t\in\R_+$ and by Lemma~\ref{lem:v} it must hold that $\int_0^\rho{[1-\delta_2 \bau_4(s)-\alpha_2 \bau_1(s)]\,\d s}>0$ and moreover, $\ubau_2(t)\ge \bau_2(t)$ for all $t\in\R_+$. Considering the evolution of $\ubau_3(t)-\bau_3(t)$ it then follows directly from Lemma~\ref{lem:w} that $\ubau_3(t)\ge \bau_3(t)$ for all $t\in\R_+$.

Consider $\ubau_1(t)$ and $\bau_1(t)$ which satisfy
$$
\ubau_1'=\ubau_1[1-\alpha_1\ubau_2(t)-\delta_1\ubau_3(t)-\ubau_1]\quad \text{and}\quad \bau_1'=\bau_1[1-\alpha_1\bau_2(t)-\bau_1].
$$
Note that $\int_0^\rho{[1-\alpha_1 \ubau_2(s)-\delta_1 \ubau_3(s)]\,\d s}>0$ and that $1-\alpha_1\ubau_2(t)-\delta_1\ubau_3(t)<1-\alpha_1\bau_2(t)$ for all $t\in\R_+$. Then from Lemma~\ref{lem:v}, $\ubau_1(t)\le \bau_1(t)$ for all $t\in\R_+$. Hence it has been shown that $M\ubu(t)\le M\bbu(t)$ for all $t\in\R_+$.

Assume that $\bu(0),\ubu(0),\bbu(0)\in \R_+$ and $M\ubu(0)\le M\bu(0) \le M\bbu(0)$. We claim that
\begin{equation}\label{eqn:claim-bounds-uvw}
M\ubu(t)\le M\bu(t) \le M\bbu(t)\quad\text{for all}\quad t\in\R_+.
\end{equation}

It is clear that $\bF_\bu',\ubF_\bu',\bbF_\bu'\in C(\R_+\times\R^4,\R^{4^2})$, hence $\bF,\bbF,\ubF$ each satisfy a local Lipschitz condition on any $D\subset\R_+\times\R^4$. It can easily be seen that $M\ubF(t,\bfeta)\le M\bF(t,\bfeta) \le M\bbF(t,\bfeta)$ for all $(t,\bfeta)\in\R_+^{1+4}$. Letting $E=\R_+\times(-\infty,0]^2\times\R_+$, we can see that the Jacobian matrices $[M\bF(t,M\bfeta)]_\bfeta'$ and $[M\bbF(t,M\bfeta)]_\bfeta'$ are essentially positive on $\R_+\times E$, meaning $M\bF(t,M\bfeta)$ and $M\bbF(t,M\bfeta)$ are quasimonotone increasing on $\R_+\times E$. Hence by Corollary~\ref{cor:ComparisonCorollary} the claim is proved true.

\begin{theorem}\label{thm:sys-coexist}
If \eqref{eqn:sys-coexist-par-con} is satisfied, then there exists a strictly positive $\rho$-periodic solution to \eqref{eqn:sys}.
\end{theorem}

\begin{proof}
Under the given parameter restrictions there exist strictly positive $\rho$-periodic solutions to \eqref{eqn:sys2} and \eqref{eqn:sys3}. Let $\ubu(t)$ and $\bbu(t)$ denote the strictly positive $\rho$-periodic solutions for systems \eqref{eqn:sys2} and \eqref{eqn:sys3}, respectively. Note that it was shown $M\ubu(t)\le M\bbu(t)$ for all $t\in\R_+$. Construct the box
$$
R=\{\bfeta\in\R_+^4: M\ubu(0)\le M\bfeta \le M\bbu(0)\}\subset \inte(\R_+^4)
$$
and define the solution of \eqref{eqn:sys} with initial condition $\bfeta\in R$ as $\bu(t;\bfeta)$, i.e. $\bu(0;\bfeta)=\bfeta$. Define the map $\varphi:R\to\R_+^4$ by
$$
\varphi(\bfeta)=\bu(\rho,\bfeta)\quad\text{for}\quad \bfeta\in R.
$$
Note from continuous dependence on initial conditions that $\varphi$ is clearly continuous. If $\bfeta\in R$, then from \eqref{eqn:claim-bounds-uvw} it is clear that
\begin{equation}\label{eqn:boundpositive}
M\ubu(t)\le M\bu(t;\bfeta)\le M\bbu(t)\quad\text{for all}\quad t\in\R_+.
\end{equation}
Then from the periodicity of $\ubu(t)$ and $\bbu(t)$ it follows that
$$
M\ubu(0)=M\ubu(\rho)\le M\bu(\rho;\bfeta)\le M\bbu(\rho)=M\bbu(0),
$$
that is, $\bu(\rho;\bfeta)\in R$ which implies $\varphi(R)\subset R$. Hence by Brouwer's Fixed Point Theorem there exists $\bfeta_0\in R$ such that $\bfeta_0=\varphi(\bfeta_0)$, i.e. $\bu(0;\bfeta_0)=\bu(\rho;\bfeta_0)$. Therefore by \citep[Lemma~2.2.1]{Farkas1994} for initial condition $\bu(0)=\bfeta_0$ the solution $\bu(t)$ must be $\rho$-periodic and from \eqref{eqn:boundpositive} that solution must be strictly positive.
\hfill\end{proof}

\subsection{Special periodic solutions of the full system}

We now classify all the special periodic solutions to system \eqref{eqn:sys} and determine the stability of each solution. With a slight abuse of notation, the special periodic solutions are of the form:
\begin{enumerate}
\item[PS1.] $(0,0,0,u_4(t))$;
\item[PS2.] $(u_1(t),0,0,u_4(t))$;
\item[PS3.] $(0,u_2(t),u_3(t),u_4(t))$;
\item[PS4.] $(u_1(t),u_2(t),u_3(t),u_4(t))$.
\end{enumerate}
Here, $u_j(t + \rho) = u_j(t) > 0$ for all $t \ge 0$ and $j=1,2,3,4$.

\subsubsection{PS1}

For PS1 we see that $u_4$ satisfies the ODE
$$
\frac{\d u_4}{\d t} = \beta_4 [i(t) - u_4].
$$
From Lemma~\ref{lem:w} we deduce that this has the $\rho$-periodic solution $u_4(t) = w(t;\beta_4 i,\beta_4)$ for all $t \in \R_+$ with initial condition given by \eqref{eqn:w-ic}.

\subsubsection{PS2}

For PS2 we see that $u_1$ and $u_2$ satisfy the system of equations
$$
\frac{\d u_1}{\d t} = u_1 (1 - u_1), \quad \frac{\d u_4}{\d t} = \beta_4 [i(t) - u_4].
$$
Since the equation for $u_1$ is autonomous, the only nontrivial periodic solution is $u_1(t) = 1$ for all $t \in \R_+$. Again from Lemma~\ref{lem:w} we conclude that the $u_4$~equation has $\rho$-periodic solution~$u_4(t) = w(t;\beta_4 i,\beta_4)$ for all $t \in \R_+$ with initial condition given by \eqref{eqn:w-ic}.

\subsubsection{PS3}

For PS3 the system of ODEs is
\begin{align*}
\frac{\d u_2}{\d t} & = \beta_2 u_2 (1 - u_2 - \delta_2 u_4),\\
\frac{\d u_3}{\d t} & = \beta_3 (u_2 - u_3),\\
\frac{\d u_4}{\d t} & = \beta_4 [i(t) - u_4 - \delta_4 u_2 u_4].
\end{align*}
It suffices to consider the reduced system \eqref{eqn:redsys} since from Lemma~\ref{lem:w}, if $u_2$ is positive and $\rho$-periodic, then a positive $\rho$-periodic solution of the $u_3$ equation is $u_3(t) = w(t;\beta_3 u_2,\beta_3)$ for all $t \in \R_+$ with initial condition given by \eqref{eqn:w-ic}. From Theorem~\ref{thm:red-sys-exist-unique} we know that there exists a unique strictly positive $\rho$-periodic solution $(u_2,u_4)$ to \eqref{eqn:redsys} if $\delta_2<1$.

\subsubsection{PS4}

For PS4 the system of ODEs is
\begin{align*}
\frac{\d u_1}{\d t} & = u_1 (1 - u_1 - \alpha_1 u_2 - \delta_1 u_3),\\
\frac{\d u_2}{\d t} & = \beta_2 u_2 (1 - u_2 - \alpha_2 u_1 - \delta_2 u_4),\\
\frac{\d u_3}{\d t} & = \beta_3 (u_2 - u_3),\\
\frac{\d u_4}{\d t} & = \beta_4 [i(t) - u_4 - \delta_4 u_2 u_4].
\end{align*}
From Theorem~\ref{thm:sys-coexist} it is known there exists a strictly positive $\rho$-periodic solution to this system if $\delta_2<1$, $1-(\alpha_1+\delta_1)\hat{u}_{2+}>0$, $1<(>)\alpha_1(1-\delta_2)$ and $1-\delta_2<(>)\alpha_2$.

\subsection{Stability of special periodic solutions to the full system}\label{sec:sys-stab}

We begin by determining the stability of PS1. Recall that PS1 is $(0,0,0,u_4(t))$, where $u_4(t) = w(t;\beta_4 i,\beta_4)$. Linearising \eqref{eqn:sys} about PS1, we obtain
$$
\frac{\d \by}{\d t} =
\begin{bmatrix}
1 & 0 & 0 & 0\\
0 & \beta_2 [1 - \delta_2 u_4(t)] & 0 & 0\\
0 & \beta_3 & -\beta_3 & 0\\
0 & -\beta_4 \delta_4 u_4(t) & 0 & -\beta_4
\end{bmatrix}
\by.
$$
This has a fundamental matrix~$\Phi(t) = (\phi_{ij}(t))_{1 \le i,j \le 4}$ given by
$$
\Phi(t) =
\begin{bmatrix}
\re^t & 0 & 0 & 0\\
0 & \re^{\beta_2 \int_0^t [1 - \delta_2 u_4(s)] \, \d s}  & 0 & \\
0 & \beta_3 \int_0^t \re^{-\beta_3 (t - s)} \phi_{22}(s) \, \d s  & \re^{-\beta_3 t} & 0\\
0 & - \beta_4 \delta_4 \int_0^t \re^{-\beta_4 (t - s)} u_4(s) \phi_{22}(s) \, \d s & 0 & \re^{-\beta_4 t}
\end{bmatrix}.
$$
Since this matrix is lower triangular, the characteristic multipliers (i.e. the eigenvalues of $\Phi(\rho)$) are
$$
\re^\rho, \quad \re^{\beta_2 \int_0^\rho [1 - \delta_2 u_4(s)] \, \d s}, \quad \re^{-\beta_3 \rho}, \quad \re^{-\beta_4 \rho}.
$$
As $| \re^\rho | > 1$, we conclude that PS1 is unstable.

Next we consider the stability of PS2, given by $(1,0,0,u_4(t))$, where $u_4(t) = w(t;\beta_4 i,\beta_4)$. Linearising \eqref{eqn:sys} about PS2 gives
$$
\frac{\d \by}{\d t} =
\begin{bmatrix}
-1 & -\alpha_1 & -\delta_1 & 0\\
0 & \beta_2 [1 - \alpha_2 - \delta_2 u_4(t)] & 0 & 0\\
0 & \beta_3 & -\beta_3 & 0\\
0 & -\beta_4 \delta_4 u_4(t) & 0 & -\beta_4
\end{bmatrix}
\by.
$$
With a slight abuse of notation, a fundamental matrix~$\Phi(t) = (\phi_{ij}(t))_{1 \le i,j \le 4}$ is
$$
\Phi(t) =
\begin{bmatrix}
\re^{-t} & -\int_0^t \re^{-(t - s)} [\alpha_1 \phi_{22}(s) + \delta_1 \phi_{32}(s)] \, \d s & -\delta_1 \int_0^t \re^{-(t - s)} \re^{-\beta_3 s} \d s & 0\\
0 &  \re^{\beta_2 \int_0^t [1 - \alpha_2 - \delta_2 u_4(s)] \, \d s} & 0 &  0 \\
0 &  \beta_3 \int_0^t \re^{-\beta_3 (t - s)} \phi_{22}(s) \, \d s & \re^{-\beta_3 t} & 0\\
0  & -\beta_4 \delta_4 \int_0^t \re^{-\beta_4 (t - s)} u_4(s) \phi_{22}(s) \, \d s & 0 & \re^{-\beta_4 t}
\end{bmatrix}.
$$
The characteristic multipliers are then
$$
\re^{-\rho}, \quad \re^{\beta_2 \int_0^\rho [1 - \alpha_2 - \delta_2 u_4(s)] \, \d s}, \quad \re^{-\beta_3 \rho}, \quad \re^{-\beta_4 \rho}.
$$
Thus the stability properties of PS2 will depend on the sign of $\int_0^\rho [1 - \alpha_2 - \delta_2 u_4(s)] \, \d s$. Note that the ODE for $u_4$ in PS2 is
$$
\frac{\d u_4}{\d t} = \beta_4 [i(t) - u_4(t)].
$$
Integrating from $0$ to $\rho$ gives
$$
\int_0^\rho u_4(s) \, \d s = \int_0^\rho i(s) \, \d s = \rho,
$$
so that
$$
\int_0^\rho [1 - \alpha_2 - \delta_2 u_4(s)] \, \d s = (1 - \alpha_2) \rho - \delta_2 \rho = \rho (1 - \alpha_2 - \delta_2).
$$
Therefore if $\alpha_2 + \delta_2 > 1$, then PS2 is stable. On the other hand, if $\alpha_2 + \delta_2 < 1$, then PS2 is unstable.

Finally, we look at the stability of PS3. Recall that PS3 is of the form~$(0,u_2(t),u_3(t),u_4(t))$, where $(u_2(t),u_4(t))$ is a periodic solution of \eqref{eqn:redsys} and $u_3(t) = w(t;\beta_3 u_2,\beta_3)$. Linearising \eqref{eqn:sys} about PS3 gives
\begin{equation}\label{eqn:sys-lin-normal-free}
\frac{\d \by}{\d t} = A(t) \by,
\end{equation}
where
$$
A(t) =
\begin{bmatrix}
1 - \alpha_1 u_2(t) - \delta_1 u_3(t) & 0 & 0 & 0\\
-\alpha_2 \beta_2 u_2(t) & \beta_2 [1 - 2 u_2(t) - \delta_2 u_4(t)] & 0 & -\beta_2 \delta_2 u_2(t)\\
0 & \beta_3 & -\beta_3 & 0\\
0 & -\beta_4 \delta_4 u_4(t) & 0 & -\beta_4 [1 + \delta_4 u_2(t)]
\end{bmatrix}.
$$
Let $\Theta(t)=(\theta_{ij}(t))_{1\le i,j\le3}$ be the fundamental matrix that satisfies $\Theta'=C(t)\Theta';~\Theta(0)=I$, where $C(t)=(a_{ij}(t))_{2\le i,j\le 4}$. Define
$$
\bfgamma(t)=\Theta(t)\int_0^t{\Theta^{-1}(s)\bfzeta(s)\,\d s},\quad \text{where}\quad \bfzeta(t)=\begin{bmatrix} a_{21}(t)\re^{\int_0^t{a_{11}(s)\,\d s}} & 0 & 0 \end{bmatrix}^{\T}.
$$
Furthermore, let $\Psi(t)=(\psi_{ij}(t))_{1\le i,j\le2}$ be a fundamental matrix solution of \eqref{eqn:red-sys-lin}, then the fundamental matrix solution of \eqref{eqn:sys-lin-normal-free} is
$$
\Phi(t) =
\begin{bmatrix}   \re^{\int_0^t{[1-\alpha_1q_2(s)-\delta_1q_3(s)]\,\d s}}  &   0                                                &  0          &   0  \\
\gamma_1(t) &  \psi_{11}(t)  &  0  & \psi_{12}(t)\\
\gamma_2(t) &  \beta_3\re^{\beta_3 t}\int_0^t{\psi_{11}(t)\re^{\beta_3 s}\,\d s}  &   \re^{-\beta_3t}   &   \beta_3\re^{\beta_3 t}\int_0^t{\psi_{12}(t)\re^{\beta_3 s}\,\d s}   \\
\gamma_3(t)        &   \psi_{21}(t)     &   0          &   \psi_{22}(t) \\
                        \end{bmatrix}.
$$
The characteristic multipliers are then $\re^{\int_0^\rho{[1-\alpha_1u_2(s)-\delta_1u_3(s)]\,\d s}}, \quad \re^{-\beta_3\rho}$ and the eigenvalues of $\Psi(\rho)$ which were shown in Section~\ref{sec:red-sys-stab} to have modulus less that 1 if $\delta_2<1$. It is clear that $|\re^{-\beta_3\rho}|<1$, hence it is required that
$$
|\re^{\int_0^\rho{[1-\alpha_1u_2(s)-\delta_1u_3(s)]\,\d s}}|<(>)1\quad \Longleftrightarrow\quad \int_0^\rho{[1-\alpha_1u_2(s)-\delta_1u_3(s)]\,\d s}<(>)0
$$
for the solution PS3 to be stable (unstable). From Lemmata~\ref{lem:w}~and~\ref{lem:x1-quadratic}
$$
\int_0^\rho{[1-\alpha_1u_2(s)-\delta_1u_3(s)]\,\d s}=\int_0^\rho{[1-(\alpha_1+\delta_1)u_2(s)]\,\d s}\ge\rho[1-(\alpha_1+\delta_1)\hat{u}_{2+}],
$$
therefore if $1-(\alpha_1+\delta_1)\hat{u}_{2+}>0$, then the solution is unstable. Now, from \eqref{eqn:redsys} it can be shown that
$$
\frac{1}{\rho}\int_0^\rho{u_2(s)\,\d s}=1-\delta_2\frac{1}{\rho}\int_0^\rho{u_4(s)\,\d s}.
$$
Furthermore, by the periodicity of $u_4$ and the positivity of $u_2$ and $u_4$, it can be obtained from \eqref{eqn:redsys} that
$$
\frac{1}{\rho}\int_0^\rho{u_4(s)\,\d s}=1-\frac{1}{\rho}\int_0^\rho{u_2(s)u_4(s)\,\d s}<1
$$
and as a result we can conclude
\begin{align}
\frac{1}{\rho}\int_0^\rho{1-\alpha_1u_2(s)-\delta_1u_3(s)\,\d s}=&\,1-(\alpha_1+\delta_1)\frac{1}{\rho}\int_0^\rho{u_2(t)\,\d s}\\
=&\,1-(\alpha_1+\delta_1)\left(1-\delta_2\frac{1}{\rho}\int_0^\rho{u_4(t)\,\d s}\right)\\
<&\,1-(\alpha_1+\delta_1)(1-\delta_2)<0
\end{align}
if $\delta_2<(\alpha_1+\delta_1-1)/(\alpha_1+\delta_1)$. Noting that $(\alpha_1+\delta_1-1)/(\alpha_1+\delta_1)<1$, it can then be concluded that if $\delta_2<(\alpha_1+\delta_1-1)/(\alpha_1+\delta_1)$, then PS3 is a stable.

\section{Discussion}\label{sec:discussion}

It is first noted that the results obtained for the model proposed by Byrne~\cite{Byrne2003}, corresponding to system \eqref{eqn:redsys}, have been further extended, that is, an examination of the behaviour of the model when $i(t)$ is an arbitrary continuous-time-periodic function has been presented. In the analysis conducted by Byrne~\cite{Byrne2003}, $i$ was given by \eqref{eqn:perinf-non-dim} and it was assumed that $x_2\equiv i$ and then \eqref{eqn:redsys} reduced to a single Bernoulli equation that could be readily solved. Here, no assumptions were made about $x_2$, rather the model was considered for all $i\in C(\R_+,[0,i_M])$, where $i(t)=i(t+\rho)$ for some $\rho>0$ and $t\in\R_+$. It was found that there exist $\rho$-periodic solutions to \eqref{eqn:redsys} that were stable for different parameter values. The trivial tumour-tissue free solution (i.e. $x_1\equiv 0$) was found to exist for all parameter values and was found to be asymptotically stable for $\delta_2>1$ and unstable for $\delta_2<1$. This is the same result observed for the analogous SS (i.e. RS1) of the model using the constant infusion function. Hence this suggests that using a different method of drug delivery will not change the conditions in which the tumour can be removed from the system, rather it is only required that the same total amount of drug is delivered over each treatment period. For system \eqref{eqn:redsys} a $\rho$-periodic solution of the form $\bx(t)=\begin{bmatrix}x_1(t) & x_2(t)\end{bmatrix}^{\T} > 0$ was found to exist for $\delta_2<1$. Furthermore, this solution was shown to be asymptotically stable if $\delta_2<1$. It is also noted that if $\delta_2>(1+\delta_4)^2/4\delta_4$, then no $\rho$-periodic solution of this form can exist, as is the case for the analogous SS of the constant infusion model (i.e. RS2). It was also shown that if $1\le \delta_2\le (1+\delta_4)^2/4\delta_4$ and $\delta_4<1$, then no biologically meaningful $\rho$-periodic solution of this form could exist. Once again, this condition is the same as for the analogous SS of the constant infusion model.

Considering the trivial periodic solution given by PS1, this solution, as in the constant infusion model, is unconditionally unstable and as a result this suggests that the model should always contain normal-tissue or tumour-tissue. This behaviour is to be expected and is consistent with cell models. The less trivial state PS2, that represents the tumour-free state, is shown to be stable when $\alpha_2+\delta_2>1$ (similarly, unstable if $\alpha_2+\delta_2<1$). Note that this is the condition for stability of the analogous SS in the constant infusion case (i.e. SS2). Whilst it should naturally follow that the periodic stability conditions should imply when the SSs of the constant infusion model are stable (respectively, unstable), it should be noted that these conditions are independent of $i(t)$ and are identical for all non-negative continuous periodic functions. Hence these conditions suggest that for the normal-tissue to remain within the system, at the least, there needs to be sufficiently strong population competition and treatment strength. Should there not be significant competition provided by the normal-tissue or large enough treatment strength cannot be obtained, for example, if the required dose to do so is unsafe, then this state will be unstable and this would provide the ideal conditions for tumour invasion. Hence if the competition that is provided by the normal-tissue is able to be increased, then this would enable the tumour-free solution to become stable and improve the potential efficacy of the treatment. It is clear from this result alone that population competition has a potentially important role to play in treating tumour invasion. Furthermore, if a treatment somehow indirectly weakens the effective competition that normal-tissue can provide without lowering the tumour-tissue competition a proportional amount, then this can actually be harmful to the potential efficacy of the treatment. In a case like this, the assessment would need to be made of whether the relative benefit gained in fighting the tumour with the specific treatment outweighs the loss incurred from the damaged competition that the normal-tissue provides.

The normal-tissue free periodic solution (i.e. PS3) was shown to exist if $\delta_2<1$ and furthermore could not exist if $\delta_2>(1+\delta_4)^2/4\delta_4$, or if $1\le\delta_2<(1+\delta_4)^2/4\delta_4$ and $\delta_4\le 1$. Hence this represents the parameter condition in which it can be assured that an invasive tumour will not exist. In this state, the concentration of chemotherapy drug is lower than the tumour-free state, as would be expected due to the model assumption that interaction with the tumour causes some portion of the drug to decay. From the stability analysis in Section~\ref{sec:sys-stab}, it can be seen that the normal-tissue free periodic solution is unstable if
$$
1-(\alpha_1+\delta_1)\hat{u}_{2+}>0.
$$
Hence there is a sufficiently large treatment strength that needs to be obtained in order prevent this invasive normal-tissue free state from being able to invade. Furthermore, from this condition it can be concluded that should $\alpha_1+\delta_1\le1$, then the state is unstable. Therefore if the combined strength of the tumour-tissue competition and the destructive influence of the acid is low, then the tumour will not be invasive. This is consistent with the results obtained for the constant infusion model considered in Section~\ref{sec:sys-constant} and for the heterogeneous model considered by McGillen~et~al.~\cite{McGillenetal2013}. This further demonstrates the potential importance of the acid-mediation hypothesis, in that, should a tumour provide low population competition, then invasion may still be achieved provided a sufficiently strong destructive influence of the acid. This once again is consistent with the results of the model proposed by McGillen~et~al.~\cite{McGillenetal2013}.

It was shown in Section~\ref{sec:sys-stab} that if $\delta_2<(\alpha_1+\delta_1-1)/(\alpha_1+\delta_1)$, then the normal-tissue free solution (i.e. PS3) is asymptotically stable. Hence for sufficiently small treatment strength and sufficiently large tumour population competition and tumour aggressiveness, the invasive tumour state will be stable. It should be noted that the normal-tissue free solution could still be stable for larger values of $\delta_2$, however the conducted analysis was unable to confirm stability for values of $\delta_2$ outside of this set of values. If the normal-tissue population competition is sufficiently low, then the invasive tumour state will be the only stable solution. This will result in the tumour successfully invading and the treatment being unsuccessful. If however the normal-tissue population competition is sufficiently strong, i.e. if $\alpha_2+\delta_2>1$, then the tumour-tissue free solution will be stable and hence the system will be bistable. Should this be the case, the size of the initial tumour will alter the efficacy of the treatment protocol which, as expected, is consistent with the results of the constant infusion model and the decreased likelihood of a cure associated with more established tumours \cite{Perry2008}. From this it can be seen that the population competition, that is, the relative interaction between different cell types, can have a significant impact on the efficacy of the tumour treatment.

It was shown that a strictly positive $\rho$-periodic coexistence solution exists for \eqref{eqn:sys} (i.e. PS4). Moreover, numerical simulations, using \eqref{eqn:perinf-non-dim} for $i$, suggest that this solution is stable for particular parameter values. It is noted that this state exists and moreover, is stable, when the tumour aggressiveness, tumour-tissue population competition, normal-tissue population competition and destructive influence of the chemotherapy is low (i.e. $\delta_1,\alpha_1,\alpha_2,\delta_2$ are ``small''). This is consistent with the results obtained for the constant infusion model. Should any of these parameters increase, the properties of the model change dramatically. If the tumour aggressiveness increases, then the tumour would become invasive as the normal-tissue free periodic solution would become stable while the tumour-tissue free solution would remain unstable. Conversely, should the destructive influence of the chemotherapy be increased, by way of increased drug infusion (say), the coexistence state would be come unstable and the tumour-free periodic solution would become stable resulting in the tumour being removed from the system.

\section{Concluding remarks}\label{sec:concluding-remarks}

A model for the acid-mediation hypothesis in the presence of a chemotherapy treatment has been proposed and considered. The proposed model is a simple ODE model that is comprised of the normal-tissue, tumour-tissue, acid concentration, and chemotherapy drug concentration in a homogeneous setting. The model was based on the model proposed by McGillen~et~al.~\cite{McGillenetal2013} in combination with that proposed by Byrne~\cite{Byrne2003} and has been considered to obtain an understanding of the reaction dynamics governing the system and to provide insights required before considering this in a heterogeneous setting. The model was considered mathematically for different treatment methods using both numerical and analytical techniques.

The model has been considered with constant drug infusion which produced an autonomous system that was studied using a steady state analysis. The model was also considered assuming the use of treatment occurring in cycles, which was characterised by time-periodic infusion functions. This resulted in a non-autonomous system that could be examined using an analysis of time-periodic solutions. The results from each analysis draw similar, if not the same, conclusions about the effect of competition, the treatment strength and the destructive influence of acid on the overall system dynamics. This suggests that the method of drug delivery is not a significant factor when trying to treat a tumour, rather it is the average rate of delivery which is the important factor. Hence much more focus can be placed on ensuring the safest method of delivery is used. Moreover, from a modelling stand point this suggests, at least in a homogeneous model, that the choice of infusion function is not as influential to the overall behaviour of the system as may be intuitively thought. This however only relates to the long term behaviour of the system, whereas the short term dynamics may still vary largely based on the choice of infusion function. Furthermore, this does not consider the potential dynamics that could be displayed in a heterogeneous setting in which spatial variation and associated mechanisms must be considered.

Since the model considered in this article assumes homogeneous populations, that is, well mixed populations, there are natural limitations to the conclusions that can be drawn from this analysis. However analysis conducted of this homogeneous model provides insights into the potential long-term behaviour of a heterogeneous version of this model, particularly in the situations of monostable solutions for the system. Hence analysis of a model that considers a heterogeneous setting will be considered in a following article to further understand the dynamics of the acid-mediation hypothesis with chemotherapy intervention.

\begin{appendix}

\section{Auxiliary definitions and results}\label{sec:appendix}

We include for the convenience of the reader a collection of definitions and results required for this article.

\begin{definition}[{Quasimonotonicity, see \citep[\S10, XII]{Walter1998}}]\label{def:quasimonotone}
The function $\bF:D\subset\R^{n+1}\to\R^n$ is said to be quasimonotone increasing on $D$ if for $i=1,\ldots,n$,
$$
\bu\le\bv,~ u_i=v_i,~(t,\bu),(t,\bv)\in D\quad \Longrightarrow \quad F_i(t,\bu)\le F_i(t,\bv).
$$
\end{definition}

\begin{remark}\label{rem:ess-pos}
A matrix $C=(c_{ij})$ is said to be \textit{essentially positive} if $c_{ij}\ge 0$ for all $i\ne j$. A function $\bF:D\subset\R^{n+1}\to\R^n$ is quasimonotone increasing on a set $D\subset \R^{n+1}$ if the Jacobian matrix $\bF_\bu'(t,\bu)$ is \textit{essentially positive} for all $(t,\bu)\in D$.
\end{remark}

\begin{definition}[{Invariant Sets in $\R^n$, see \citep[\S10, XV]{Walter1998}}]\label{def:invariantset}
A set $D\subset \R^n$ is said to be invariant with respect to the system $\bu'=\bF(t,\bu)$ if for any solution $\bu$, $\bu(t_0)\in D$ implies $\bu(t)\in D$ for $t>t_0$ (as long as the solution exists).
\end{definition}

\begin{definition}[{Tangent Condition, see \citep[\S10, XV]{Walter1998}}]\label{def:tangentcondition}
Let $D\subset\R^n$ and $\bF:E\supset J\times \overline{D}\to \R^n$, where $J\subset\R$ is an interval. The tangent condition is given by
\begin{equation}
\inner{\bn(\bu)}{\bF(t,\bu)}\le 0,\quad \text{for}\quad t\in J,~\bu\in\partial D,
\end{equation}
where $\bn(\bu)$ is the outer normal to $D$ at $\bu$.
\end{definition}

\begin{definition}[{Local Lipschitz Condition, see \citep[\S6, IV]{Walter1998}}]\label{def:locallipschitz}
Let $J\subset \R$, $D\subset\R^n$ and $E=J\times D$. Then the function $\bF:E\to \R^n$ is said to satisfy a local Lipschitz condition with respect to $\bu$ in $E$ if for every $(t_0,\bu_0)\in E$ there exists a neighbourhood $U=U(t_0,\bu_0)$ and an $L=L(t_0,\bu_0)<\infty$ such that for all $(t,\bu_1),(t,\bu_2)\in U\cap E$ the function $\bF$ satisfies the Lipschitz condition
\begin{equation}\label{eqn:LocalLipschitz}
\|\bF(t,\bu_1)-\bF(t,\bu_2)\|\le L\|\bu_1-\bu_2\|.
\end{equation}
\end{definition}

\begin{remark}
If $D$ is open and if $\bF\in C(E,\R^n)$ has continuous derivative $\bF_\bu'(t,\bu)$ in $E$, then $\bF$ satisfies a local Lipschitz condition with respect to $\bu$ in $E$.
\end{remark}

\begin{theorem}[{Invariance Theorem, see \citep[\S10, XVI]{Walter1998}}]\label{thm:Invariance-Exist-Uniqe}
Let $D\subset\R^n$ be closed, $\bF:[t_0,\infty)\times D\to \R^n$ bounded and continuous and consider the system $\bu'=\bF(t,\bu)$. Suppose that $\bF$ satisfies the tangent condition~(see Definition~\ref{def:tangentcondition}) and the one-sided Lipschitz condition on $D$, that is
\begin{equation}\label{eqn:rightLip}
\inner{\bu_1-\bu_2}{\bF(t,\bu_1)-\bF(t,\bu_2)}\le L \|\bu_1-\bu_2\|^2\quad\text{for all}\quad t\in[t_0,\infty),~\bu_1,\bu_2\in D.
\end{equation}
Then for any $\bu(t_0)\in D$ a unique solution $\bu(t)$ exists for all $t\in[t_0,\infty)$ which is invariant on $D$.
\end{theorem}

\begin{theorem}[{Comparison Theorem, see \citep[\S10, Comparison Theorem]{Walter1998}}]\label{thm:ComparisonTheorem}
Let $D\subset \R^{n}$ and $J=[t_0,t_0+a]$; assume that $\bF:J\times D\to \R^n$ is quasimonotone increasing and that $\bF(t,\bu)$ satisfies a local Lipschitz condition with respect to $\bu$ in $J\times D$. Suppose that $\bu'=\bF(t,\bu)$ and $\bv'\le\bF(t,\bv)$, if $\bu(t)$ and $\bv(t)$ are differentiable on $J$ and $\bv(t_0)\le \bu(t_0)$, then $\bv(t)\le \bu(t)$ for all $t\in J$.
\end{theorem}

As a consequence of Theorem~\ref{thm:ComparisonTheorem} we have the following corollary.
\begin{corollary}\label{cor:ComparisonCorollary}
Let $D\subset \R^{n}$, $J=[t_0,t_0+a]$ and $\bF:J\times D\to \R^n$; assume that $\bF(t,\bu)$ satisfies a local Lipschitz condition with respect to $\bu$ on $J\times D$ and that $\bu$ satisfies $\bu'=F(t,\bu)$. Suppose that there exists an invertible $n\times n$ matrix $M$ such that $M\bv' \le M\bF(t,\bv)$ and $M\bF(t,M^{-1}\bfeta)$ is quasimonotone increasing on $J\times \{\bfeta\in\R^n:M^{-1}\bfeta\in D\}$. If $\bu(t)$ and $\bv(t)$ are differentiable on $J$ and $M\bv(t_0) \le M \bu(t_0)$, then $M\bv(t) \le M\bu(t)$ for all $t\in J$.
\end{corollary}

\begin{proof}
It is well known that for any linear operator $T:\Omega\subset\R^n\to\R^n$ there exists a constant $c\in\R$ such that $\|T\bx\|\le c\|\bx\|$ for all $\bx\in \Omega$ \citep[p. 58]{Walter1998}. Hence there exists $c\in\R$ such that $\|M\bF(t,\bu_1)-M\bF(t,\bu_2)\|=\|M(\bF(t,\bu_1)-\bF(t,\bu_2))\|\le c\|\bF(t,\bu_1)-\bF(t,\bu_2)\|$ for all $(t,\bu_1),(t,\bu_2)\in J\times D$. Then as $\bF(t,\bu)$ satisfies a local Lipschitz condition on $J\times D$ it easily follows that so too does $M\bF(t,\bu)$.

Let $\bbu(t)=M\bu(t)$ and $\bbv(t)=M\bv(t)$, then we have $\bbu'=M\bF(t,M^{-1}\bbu)$ and $\bbv'\le M\bF(t,M^{-1}\bbv)$ on $J\times E$, where $E=\{\bfeta\in\R^n:M^{-1}\bfeta\in D\}$. Note $M\bF(t,M^{-1}\bbu)$ is quasimonotone increasing on $J\times E$ and if $\bu(t)$ and $\bv(t)$ are differentiable on $J$, so too are $\bbu(t)$ and $\bbv(t)$. Therefore by Theorem~\ref{thm:ComparisonTheorem} if $\bbv(t_0) \le \bbu(t_0)$, then $\bbv(t) \le \bbu(t)$ for all $t\in J$, that is, if $M\bv(t_0) \le M \bu(t_0)$, then $M\bv(t) \le M\bu(t)$ for all $t\in J$.
\hfill\end{proof}

\section{Results for specific ODEs and corresponding periodic solutions}\label{sec:appendix2}

\begin{lemma}
\label{lem:w}
Consider the equation
\begin{equation}
\label{w-ode}
\frac{\d w}{\d t} = f(t) - g(t) w,
\end{equation}
where $f, g \in C(\R_+)$ are $\rho$-periodic. Suppose that
$$
\int_0^\rho{g(s)\,\d s}\ne 0,\quad f(t) \ge(>) 0 \quad \text{for all}\quad t\in\R_+.
$$
Then the following statements hold:
\begin{enumerate}
\item[\rm (i)] The function
\begin{equation}
\label{w-sol}
w(t;f,g) = w(0) \re^{-\int_0^t g(s) \, \d s} + \int_0^t \re^{-\int_s^t g(s') \, \d s'} f(s) \, \d s
\end{equation}
is the unique solution of \eqref{w-ode} for any initial condition~$w(0)\in \R$. Note that $w(t;f,g) > 0$ for all $t \ge 0$ if and only if $w(0) > 0$.
\item[\rm (ii)] If
\begin{equation}
\label{eqn:w-ic}
w(0) = \frac{\int_0^\rho \re^{-\int_s^\rho g(s') \, \d s'} f(s) \, \d s}{1 - \re^{-\int_0^\rho g(s) \, \d s}},
\end{equation}
then \eqref{w-sol} is the unique $\rho$-periodic solution, where
\begin{equation}
\label{w-int}
\int_0^\rho g(t) w(t) \, \d t = \int_0^\rho f(t) \, \d t.
\end{equation}
Moreover, $w(0)\ge(>)0$ if and only if $\int_0^\rho{g(s)\,\d s}>0$.

\item[\rm (iii)] Suppose that $g_1(t)\ge g_2(t)$ for all $t\in\R_+$ and $\int_0^\rho{g_2(s)\,\d s}>0$. If $w(t;f,g_1)$ and $w(t;f,g_2)$ have initial conditions given by \eqref{eqn:w-ic}, then $w(t;f,g_1)\le w(t;f,g_2)$ for all $t\in\R_+$.
\end{enumerate}
\end{lemma}
\begin{proof}
\begin{enumerate}
\item[\rm (i)]
The unique solution~$w = w(\cdot;f,g)$ of the linear ODE~\eqref{w-ode} for any initial condition is given by \eqref{w-sol}. Since $f(t)\ge 0$ for all $t\in\R_+$ it is straightforward to see that $w(t;f,g) > 0$ for all $t \ge 0$ if and only if $w(0) > 0$.
\item[\rm (ii)]
If the initial condition is given by \eqref{eqn:w-ic}, then $w(\rho) = w(0)$. Furthermore, the $\rho$-periodicity of $f$ and $g$ implies that
$$
w'(t) = f(t)-g(t) w(t), \quad w'(t + \rho) = f(t)-g(t) w(t + \rho) \quad \text{for all $t \ge 0$}.
$$
Then $\bar w(t) = w(t + \rho) - w(t)$ satisfies the initial value problem
$$
\frac{\d \bar w}{\d t} = -g(t) \bar w, \quad \bar w(0) = 0.
$$
Hence $\bar w \equiv 0$ and $w(t + \rho) = w(t)$ for all $t \ge 0$. Since $f(t)\ge(>)0$ for all $t\in\R_+$ it is clear that if $w(0)$ is given by \eqref{eqn:w-ic}, then $w(0)\ge(>)0$ if and only if $\int_0^\rho{g(s)\,\d s}>0$. To show that the solution for \eqref{eqn:w-ic} is the unique $\rho$-periodic solution we consider two distinct $\rho$-periodic solutions for \eqref{w-ode} denoted by $w_1(t)$ and $w_2(t)$. Due to uniqueness to initial conditions $w_1(t)\ne w_2(t)$ for all $ t\in\R_+$, so without loss of generality we assume $w_1(t)>w_2(t)$ for all $t\in \R_+$. Then $\varepsilon(t)=w_1(t)-w_2(t)=w_1(t+\rho)-w_2(t+\rho)=\varepsilon(t+\rho)>0$ for all $t\in\R_+$ satisfies
$$
\Der{\varepsilon}{t}=-g(t)\varepsilon\quad\Longrightarrow\quad 0=\int_0^\rho{\frac{1}{\varepsilon(s)}\Der{\varepsilon}{s}\,\d s}=-\int_0^\rho{g(s)\,\d s},
$$
which is a contradiction. Therefore the $\rho$-periodic solution for initial condition \eqref{eqn:w-ic} must be unique.
\item[\rm (iii)]
Suppose that $g_1(t)\ge g_2(t)$ for all $t\in\R_+$ and $\int_0^\rho{g_2(s)\,\d s}>0$, then $\int_0^\rho{g_1(s)\,\d s}\ge \int_0^\rho{g_2(s)\,\d s}>0$. If $w(t;f,g_1)$ and $w(t;f,g_2)$ have initial conditions given by \eqref{eqn:w-ic}, then $w(t;f,g_1)$ and $w(t;f,g_2)$ are unique $\rho$-periodic solutions and $w(t;f,g_1),w(t;f,g_2)\ge 0$ for all $t\in\R_+$. Let $\varepsilon(t)=w(t;f,g_2)-w(t;f,g_1)=w(t+\rho;f,g_2)-w(t+\rho;f,g_1)=\varepsilon(t+\rho)$ which satisfies
$$
\varepsilon'=[g_1(t)-g_2(t)]w(t;f,g_2)-g_1(t)\varepsilon.
$$
Now $[g_1(t)-g_2(t)]w(t;f,g_2)\ge 0$ for all $t\in\R_+$, hence by the previous results the $\rho$-periodic solution $\varepsilon(t)\ge 0$ for all $t\in \R_+$, i.e. $w(t;f,g_1)\le w(t;f,g_2)$ for all $t\in\R_+$.
\end{enumerate}
\hfill\end{proof}

\begin{lemma}
\label{lem:v}
Consider the equation
\begin{equation}
\label{v-ode}
\frac{\d v}{\d t} = g(t) v - f(t) v^2,
\end{equation}
where $f, g \in C(\R_+)$ are $\rho$-periodic. Suppose that
$$
\int_0^\rho g(s) \, \d s \ne 0, \quad f(t) > 0 \quad \text{for all}\quad t\in\R_+.
$$
Then the following statements hold:
\begin{enumerate}
\item[\rm (i)] The function
\begin{equation}
\label{v-sol}
v(t;f,g) = \left[v(0)^{-1} \re^{-\int_0^t g(s) \, \d s} + \int_0^t \re^{-\int_s^t g(s') \, \d s'} f(s) \, \d s\right]^{-1}
\end{equation}
is the unique solution of \eqref{v-ode} for any initial condition~$v(0) > 0$. Note that $v(t;f;g) > 0$ for all $t \in \R_+$ if and only if $v(0) > 0$.
\item[\rm (ii)] If $\int_0^\rho{g(s)\,\d s>0}$ and
\begin{equation}
\label{eqn:v-ic}
v(0) = \frac{1-\re^{-\int_0^\rho g(s') \, \d s'}}{\int_0^\rho \re^{-\int_s^\rho g(s') \, \d s'} f(s) \, \d s},
\end{equation}
then \eqref{v-sol} is a unique strictly positive $\rho$-periodic solution. Moreover, $v(0)>0$ if and only if $\int_0^\rho{g(s)\,\d s>0}$.
\item[\rm (iii)] Suppose that $g_1(t)\ge g_2(t)$ for all $t\in\R_+$ and $\int_0^\rho{g_2(s)\,\d s}>0$. If $v(t;f,g_1)$ and $v(t;f,g_2)$ have initial conditions given by \eqref{eqn:v-ic}, then $v(t;f,g_1)\ge v(t;f,g_2)$ for all $t\in\R_+$.
\end{enumerate}

\end{lemma}
\begin{proof}
\begin{enumerate}
\item[\rm (i)]
Set $v = v(\cdot;f,g)$ in the Bernoulli equation~\eqref{v-ode} to $v(t) = u(t)^{-1}$. Then $u$ satisfies the linear ODE
$$
\frac{\d u}{\d t} = f(t) - g(t) u,
$$
whose solution for any initial condition is
$$
u(t) = u(0) \re^{-\int_0^t g(s) \, \d s} + \int_0^t \re^{-\int_s^t g(s') \, \d s'} f(s) \, \d s.
$$
Hence \eqref{v-sol} follows.
\item[\rm (ii)]
Since $f(t)>0$ for all $t\in\R_+$ it is clear that if $v(0)$ is given by \eqref{eqn:v-ic}, then $v(0)>0$ if and only if $\int_0^\rho{g(s)\,\d s}>0$. If $\int_0^\rho{g(s)\,\d s}>0$ and the initial condition is given by \eqref{eqn:v-ic}, then $u(\rho)=u(0)>0$ and a similar analysis as in the proof of Lemma~\ref{lem:w} shows that $v$ is a unique strictly positive $\rho$-periodic solution.
\item[\rm (iii)]
Suppose that $g_1(t)\ge g_2(t)$ for all $t\in\R_+$ and $\int_0^\rho{g_2(s)\,\d s}>0$. If the initial conditions for $v(t;f,g_1)$ and $v(t;f,g_2)$ are given by \eqref{eqn:v-ic}, then Lemma~\ref{lem:w}(iii) shows that $u(t;f,g_1)\le u(t;f,g_2)$ for all $t\in \R_+$, which implies $v(t;f,g_1)\ge v(t;f,g_2)$ for all $t\in\R_+$.
\end{enumerate}
\hfill\end{proof}

\begin{lemma}\label{lem:x1-quadratic}
Suppose that there exists a solution $\bx(t)=\begin{bmatrix} x_1(t) & x_2(t)\end{bmatrix}^{\T}>0$ for all $t\in\R_+$ to system \eqref{eqn:redsys}, then it follows that
\begin{equation}\label{eqn:redquadratic}
\delta_4 x_1(t)^2+(1-\delta_4)x_1(t)+\delta_1 i(t)-1=\Der{}{t}V(\bx(t)),
\end{equation}
where $V(\bx)=-(\delta_4/\beta_2)x_1+(\delta_2/\beta_4)x_2-(1/\beta_2)\ln x_1$.  Moreover, it is a necessary condition that
$$
\delta_2 i(t) -\tfrac{\d}{\d t}V(\bx(t))\le \frac{(1+\delta_4)^2}{4\delta_4}.
$$
If the solution is $\rho$-periodic, then
\begin{enumerate}
\item[\rm (i)] The solution exists only if $\delta_2<1$ or only if $1\le \delta_2 \le (1+\delta_4)^2/4\delta_4$ and $\delta_4 > 1$. Moreover,
    $$
    0\le\frac{1}{\rho}\int_0^\rho{x_1(s)\,\d s}\le \hat{u}_{2+}.
    $$

\item[\rm (ii)] If $\delta_2<1$, then
\begin{equation}\label{eqn:x1sol}
x_1(t)=\frac{\delta_4-1+\sqrt{(\delta_4+1)^2-4\delta_4[\delta_2 i(t)-h(t)]}}{2\delta_4},
\end{equation}
where $h(t)=\tfrac{\d}{\d t}V(\bx(t))$.
\end{enumerate}
\end{lemma}
\begin{proof}
We wish to find a function~$V = V(\bx)$ such that
$$
\frac{\d}{\d t} V(\bx(t)) = A(t) x_1(t)^2 + B(t) x_1(t) + C(t),
$$
where $A$, $B$ and $C$ are appropriate $\rho$-periodic functions. Note the right-hand side is independent of $x_2$. By considering \eqref{eqn:redsys}, we can try the ansatz
$$
V(\bx) = a x_1 + b x_2 + c \log x_1,
$$
where $a$, $b$ and $c$ are constants. Then
\begin{align*}
\frac{\d}{\d t} V(\bx(t))=  & \,\left[a + \frac{c}{x_1(t)}\right] x_1'(t) + b x_2'(t)\\
=& \, -a \beta_2 x_1(t)^2 + (a - c) \beta_2 x_1(t) + c \beta_2 + b \beta_4 i(t)\\
&\,-(a \beta_2 \delta_2 + b \beta_4 \delta_4) x_1(t) x_2(t) - (c \beta_2 \delta_2 + b \beta_4) x_2(t).
\end{align*}
To eliminate the terms involving $x_2$, we set
$$
a = -\frac{\delta_4}{\beta_2}, \quad b = \frac{\delta_2}{\beta_4}, \quad c = -\frac{1}{\beta_2}.
$$
This gives
\begin{equation}\label{eqn:x1quad}
\frac{\d}{\d t} V(\bx(t)) = \delta_4 x_1(t)^2 + (1 - \delta_4) x_1(t) + \delta_2 i(t) - 1,
\end{equation}
i.e. $A(t) = \delta_4$, $B(t) = 1 - \delta_4$ and $C(t) = \delta_2 i(t) - 1$. Rewriting \eqref{eqn:x1quad}, we obtain
\begin{align*}
x_1(t) =&\, \frac{(\delta_4 - 1) \pm \sqrt{(\delta_4 - 1)^2 + 4 \delta_4 [1 - \delta_2 i(t) + h(t)]}}{2 \delta_4}\\
=&\,\frac{(\delta_4 - 1) \pm \sqrt{(\delta_4 + 1)^2 - 4 \delta_4 [\delta_2 i(t) - h(t)]}}{2 \delta_4},
\end{align*}
where $h(t) = \tfrac{\d}{\d t} V(u_2(t),u_4(t))$. If $x_1(t)>0$ for all $t\in\R_+$, a necessary condition is that
\begin{equation}\label{eqn:x1-cond}
\delta_2 i(t) - h(t) \le \frac{(\delta_4 +1)^2}{4 \delta_4}.
\end{equation}

Suppose that the solution $\bx(t)>0$ is $\rho$-periodic:
\begin{enumerate}
\item[\rm (i)]
Integrating both sides of \eqref{eqn:x1-cond} with respect to $t$ from $0$ to $\rho$ and using the periodicity of $\bx$, we have
\begin{align*}
\frac{(\delta_4 + 1)^2}{4 \delta_4}  & \ge \frac{1}{\rho} \int_0^\rho \left[\delta_2 i(t) - h(t)\right] \, \d t = \delta_2.
\end{align*}
Moreover, by the Mean Value Theorem, there exists $t^* \in (0,\rho)$ such that
$$
\frac{1}{\rho} \int_0^\rho \left[\delta_2 i(t) - h(t)\right] \, \d t = \delta_2 i(t^*) - h(t^*).
$$
We therefore deduce that
$$
x_1(t^*)=\frac{(\delta_4 - 1) \pm \sqrt{(\delta_4 + 1)^2 - 4 \delta_4\delta_2}}{2 \delta_4},
$$
which has at least one strictly positive value if and only if $\delta_2<1$, or if $1\le\delta_2\le(1+\delta_4)^2/4\delta_4$ and $\delta_4>1$, hence if these parameter conditions are not satisfied, then it is not possible for $x_1(t)>0$ for all $t\in\R_+$.

Moreover, the Cauchy-Schwarz inequality gives
\begin{align*}
\frac{1}{\rho}\int_0^\rho{\sqrt{(\delta_4 + 1)^2 - 4 \delta_4 [\delta_2 i(s) - h(s)]}\,\d s}\le &\, \left[\frac{1}{\rho}\int_0^\rho{\{(\delta_4 + 1)^2 - 4 \delta_4 [\delta_2 i(s) - h(s)]\}\,\d s}\right]^{1/2}\\
=&\,\sqrt{(\delta_4 + 1)^2 - 4 \delta_4\delta_2}.
\end{align*}
Hence
\begin{align*}
0\le\frac{1}{\rho}\int_0^\rho{x_1(s)\,\d s}\le &\, \max\left\{\frac{1}{\rho}\int_0^\rho{\frac{(\delta_4 - 1) \pm \sqrt{(\delta_4 + 1)^2 - 4 \delta_4 [\delta_2 i(s) - h(s)]}}{2 \delta_4}\,\d s}\right\}\\
\le &\, \hat{u}_{2+}
\end{align*}

\item[\rm (ii)] If $\delta_2 < 1$, we further deduce that
$$
(\delta_4 - 1)^2 + 4 \delta_4 [1 - \delta_2 i(t^*) + h(t^*)] > (\delta_4 - 1)^2
$$
and as a result
$$
x_1(t) = \frac{(\delta_4 - 1) + \sqrt{(\delta_4 - 1)^2 + 4 \delta_4 [1 - \delta_2 i(t) + h(t)]}}{2 \delta_4},
$$
otherwise $x_1$ can become negative at $t^*$.
\end{enumerate}
\hfill\end{proof}

\section{Details of steady-state analysis}\label{sec:appendix3}

\begin{lemma}\label{lem:SSfullsystem}
The system of equations given by \eqref{eqn:sys-const} has the following SS solutions and respective linear stability conditions:
\begin{enumerate}
\item[\textnormal{SS1.}]\label{con:1SS}
$\bu^*=(0,0,0,1)$ is unconditionally unstable.
\item[\textnormal{SS2.}]\label{con:2SS}
$\bu^*=(1,0,0,1)$ is stable if and only if $\alpha_2+\delta_2>1$.
\item[\textnormal{SS3.}]\label{con:3SS}
$\bu^*=(0,\hat{u}_{2\pm},\hat{u}_{2\pm},[1+\delta_4\hat{u}_{2\pm}]^{-1})$, where $\hat{u}_{2\pm}=\left[\delta_4-1\pm\sqrt{(1+\delta_4)^2-4\delta_2\delta_4}\right]/2\delta_4$. The SS corresponding to $\hat{u}_{2-}$ is unconditionally unstable and the SS corresponding to $\hat{u}_{2+}$ will be stable if and only if
$$
\delta_2<\frac{(\alpha_1+\delta_1+\delta_4)(\alpha_1+\delta_1-1)}{(\alpha_1+\delta_1)^2},
$$
or if
$$
\frac{(\alpha_1+\delta_1+\delta_4)(\alpha_1+\delta_1-1)}{(\alpha_1+\delta_1)^2}\le \delta_2 < \frac{(1+\delta_4)^2}{4\delta_4}\quad\text{and}\quad \frac{1}{\delta_4}<\frac{\alpha_1+\delta_1-2}{\alpha_1+\delta_1}.
$$
\item[\textnormal{SS4.}]\label{con:4SS}
$\bu^*=(1-[\alpha_1+\delta_1]\tilde{u}_{2\pm},\tilde{u}_{2\pm},\tilde{u}_{2\pm},[1+\delta_4\tilde{u}_{2\pm}]^{-1})$, where
\begin{equation}
\begin{split}
\tilde{u}_{2\pm}=&\,\frac{-(1-\alpha_2(\alpha_1+\delta_1))-\delta_4(\alpha_2-1)}{2\delta_4(1-\alpha_2(\alpha_1+\delta_1))}\\
&\quad\pm\frac{\sqrt{[1-\alpha_2(\alpha_1+\delta_1)-\delta_4(\alpha_2-1)]^2-4\delta_2\delta_4(1-\alpha_2(\alpha_1+\delta_1))}}
{2\delta_4(1-\alpha_2(\alpha_1+\delta_1))}.
\end{split}
\end{equation}
The SS corresponding to $\tilde{u}_{2-}$ has no biologically meaningful values for which it is stable. The SS corresponding to $\tilde{u}_{2+}$ is stable if and only if
\begin{equation}\label{eqn:vtildnecessary}
0<(\alpha_1+\delta_1)\tilde{u}_{2+}<1,\quad 0<\delta_2<\frac{[1-\alpha_2(\alpha_1+\delta_1)-\delta_4(\alpha_2-1)]^2}{4\delta_4[1-\alpha_2(\alpha_1+\delta_1)]}
\end{equation}
and
$$
c_1c_2c_3>c_3^2c_0+c_1^2,
$$
where
\begin{align}
c_3=&\,1-(\alpha_1+\delta_1)\tilde{u}_2+\beta_3 +\beta_2\tilde{u}_2+\beta_4(1+\delta_4\tilde{u}_2),\\
c_2=&\,\beta_3[\beta_2\tilde{u}_2+\beta_4(1+\delta_4\tilde{u}_2)]+[1-(\alpha_1+\delta_1)\tilde{u}_2][\beta_3+\beta_4(1+\delta_4\tilde{u}_2)+\beta_2\tilde{u}_2(1-\alpha_2\alpha_1)]\\
&\quad+\frac{\beta_2\beta_4\tilde{u}_2}{1+\delta_4\tilde{u}_2}[(1+\delta_4\tilde{u}_2)^2-\delta_2\delta_4],\\
c_1=&\,[1-(\alpha_1+\delta_1)\tilde{u}_2]\left[\frac{\beta_2\beta_4\tilde{u}_2}{1+\delta_4\tilde{u}_2}[(1+\delta_4\tilde{u}_2)^2(1-\alpha_2\alpha_1)-\delta_2\delta_4]+\beta_3\beta_4(1+\delta_4\tilde{u}_2)\right]\\
&\quad+\beta_2\beta_3\tilde{u}_2[1-(\alpha_1+\delta_1)\tilde{u}_2][1-\alpha_2(\alpha_1+\delta_1)]+\frac{\beta_2\beta_3\beta_4\tilde{u}_2}{1+\delta_4\tilde{u}_2}[(1+\delta_4\tilde{u}_2)^2-\delta_2\delta_4],\\
c_0=&\,\frac{\beta_2\beta_3\beta_4\tilde{u}_2}{1+\delta_4\tilde{u}_2}\left[1-(\alpha_1+\delta_1)\tilde{u}_2\right]\left[(1+\delta_4\tilde{u}_2)^2(1-\alpha_2(\alpha_1+\delta_1))-\delta_2\delta_4\right].
\end{align}
\end{enumerate}
\end{lemma}

\begin{proof}
We consider the system of equations given by \eqref{eqn:sys-const}, set the derivatives to zero (i.e $\bu'=\bzero$) and solve for $u_1,u_2,u_3,u_4$ to obtain the SS solutions. We obtain the solutions
$$
(0,0,0,1), (1,0,0,1), \left(0,\hat{u}_2,\hat{u}_2,\frac{1}{1+\delta_4 \hat{u}_2}\right),\left(1-(\alpha_1+\delta_1)\tilde{u}_2,\tilde{u}_2,\tilde{u}_2,\frac{1}{1+\delta_4 \tilde{u}_2}\right),
$$
where $\hat{u}_2$ solves
\begin{equation}\label{eqn:vbar}
\delta_4 \hat{u}_2^2+(1-\delta_4)\hat{u}_2+\delta_2-1=0
\end{equation}
and $\tilde{u}_2$ solves
\begin{equation}\label{eqn:vtilde}
\delta_4[1-\alpha_2(\alpha_1+\delta_1)]\tilde{u}_2^2+[1-\alpha_2(\alpha_1+\delta_1)+\delta_4(\alpha_2-1)]\tilde{u}_2+\delta_2+\alpha_2-1=0.
\end{equation}

We wish to determine the stability of these solutions by performing a linear stability analysis. Consider the following vector:
\begin{equation}
\bF(\bu)=\begin{bmatrix} u_1(1-u_1-\alpha_1 u_2-\delta_1 u_3)\\ \beta_2 u_2(1-u_2-\alpha_2 u_1-\delta_2 u_4) \\ \beta_3(u_2-u_3) \\ \beta_4(1-u_4-\delta_4u_4u_2)\end{bmatrix}.
\end{equation}
The Jacobian matrix of $\bF$ is then given by
\begin{equation}
\bF'_\bu(\bu)=\begin{bmatrix}   1-2u_1-\alpha_1 u_2-\delta_1 u_3  &   -\alpha_1 u_1                         &   -\delta_1 u_1 &   0   \\
                                        -\beta_2\alpha_2 u_2          &   \beta_2(1-2u_2-\alpha_2u_1-\delta_2u_4)   &   0           &   -\beta_2\delta_2u_2 \\
                                        0                           &   \beta_3                             &   -\beta_3    &   0   \\
                                        0                           &   -\beta_4 \delta_4 u_4                 &   0           &   -\beta_4(1+\delta_4 u_2) \\
                        \end{bmatrix}
\end{equation}
\begin{enumerate}
\item[\textnormal{SS1.}]
Now if we consider the SS solution $(0,0,0,1)$ we have the Jacobian matrix
\begin{equation}
\bF'_\bu(0,0,0,1)=\begin{bmatrix}   1  &   0                     &   0          &   0   \\
                                        0  &   \beta_2(1-\delta_2)   &   0          &   0 \\
                                        0  &   \beta_3               &   -\beta_3   &   0   \\
                                        0  &   -\beta_4 \delta_4     &   0          &   -\beta_4 \\
                        \end{bmatrix}
\end{equation}
which has eigenvalues $\lambda=1,\beta_2(1-\delta_2),-\beta_3,-\beta_4$. Hence we can see this is unstable for all parameter values as $\lambda=1>0$.
\item[\textnormal{SS2.}]
The SS $(1,0,0,1)$ has Jacobian matrix
\begin{equation}
\bF'_\bu(1,0,0,1)=\begin{bmatrix}   -1  &   -\alpha_1                      &   -\delta_1 &   0   \\
                                        0   &   \beta_2(1-\alpha_2-\delta_2)   &   0         &   0 \\
                                        0   &   \beta_3                        &   -\beta_3  &   0   \\
                                        0   &   -\beta_4 \delta_4              &   0         &   -\beta_4 \\
                        \end{bmatrix}
\end{equation}
with eigenvalues $\lambda=-1,\beta_2(1-\alpha_2-\delta_2),-\beta_3,-\beta_4$. Therefore we can see that all $\text{Re}(\lambda)<0$ if and only if $\alpha_2+\delta_2>1$. Hence we have that $(1,0,0,1)$ is linearly stable if $\alpha_2+\delta_2>1$.
\item[\textnormal{SS3.}]
We consider the SS solution $\bu^*=(0,\hat{u}_2,\hat{u}_2,[1+\delta_4\hat{u}_2]^{-1})$ noting from $\bF(\bu)=\bzero$ we have $1-u_2-\alpha_2 u_1-\delta_2 u_4=0$. Hence we have the Jacobian matrix
\begin{equation}
\bF'_\bu(\bu^*)=
\begin{bmatrix} 1-(\alpha_1+\delta_1)\hat{u}_2  &   0                                           &   0           &   0   \\
                -\beta_2\alpha_2 \hat{u}_2      &   -\beta_2\hat{u}_2                             &   0           &   -\beta_2\delta_2\hat{u}_2 \\
                0                             &   \beta_3                                     &   -\beta_3    &   0   \\
                0                             &   -\beta_4\delta_4(1+\delta_4 \hat{u}_2)^{-1}   &   0           &   -\beta_4(1+\delta_4 \hat{u}_2) \\
\end{bmatrix}
\end{equation}
with eigenvalues that satisfy $\lambda=1-(\alpha_1+\delta_1)\hat{u}_2,-\beta_3$ and $\lambda^2+[\beta_2\hat{u}_2+\beta_4(1+\delta_4\hat{u}_2)]\lambda+\beta_2\beta_4\hat{u}_2[(1+\delta_4\hat{u}_2)^2-\delta_4\delta_2]/(1+\delta_4\hat{u}_2)=0$. Therefore using the Routh--Hurwitz conditions \citep[][pp.~507--509]{MurrayI2002}, we have $\text{Re}(\lambda)<0$ if $1-(\alpha_1+\delta_1)\hat{u}_2<0$, $\delta_4\delta_2<(1+\delta_4\hat{u}_2)^2$ and $\hat{u}_2>0$. Note that if $1-(\alpha_1+\delta_1)\hat{u}_2<0$, then it follows that $\hat{u}_2>0$. Since $\hat{u}_2$ satisfies \eqref{eqn:vbar}, we have that $\delta_2=1-\delta_4\hat{u}_2^2+(\delta_4-1)\hat{u}_2$ and as a result
$$
(1+\delta_4 \hat{u}_2)^2-\delta_4\delta_2=(1+\delta_4\hat{u}_2)(2\delta_4 \hat{u}_2+1-\delta_4).
$$
Therefore if $2\delta_4 \hat{u}_2+1-\delta_4>0$, then it follows that $\delta_4\delta_2<(1+\delta_4 \hat{u}_2)^2$. Consider \eqref{eqn:vbar}, solving for $\hat{u}_2$ we obtain
$$
\hat{u}_{2\pm}=\frac{\delta_4-1\pm\sqrt{(1-\delta_4)^2+4\delta_4(1-\delta_2)}}{2\delta_4}=\frac{\delta_4-1\pm\sqrt{(1+\delta_4)^2-4\delta_4\delta_2}}{2\delta_4}
$$
and note that $\hat{u}_{2\pm}\in\R$ if and only if $\delta_2\le(1+\delta_4)^2/4\delta_4$.
Therefore
$$
2\delta_4 \hat{u}_{2\pm}+1-\delta_4=\pm\sqrt{(1+\delta_4)^2-4\delta_4\delta_2}
$$
and as a result we can see that if $\delta_2<(1+\delta_4)^2/4\delta_4$, then $2\delta_4 \hat{u}_{2+}+1-\delta_4>0$ and $2\delta_4 \hat{u}_{2-}+1-\delta_4<0$. Hence SS3 with $\hat{u}_{2-}$ will be unstable for all parameter values and we only require that $1-(\alpha_1+\delta_1)\hat{u}_{2+}<0$ for SS3 with $\hat{u}_{2+}$ to be linearly stable.
Consider
\begin{align}
&(\alpha_1+\delta_1)\hat{u}_{2+}-1\\
&\quad=\frac{(\alpha_1+\delta_1)(\delta_4-1)-2\delta_4+(\alpha_1+\delta_1)\sqrt{(1-\delta_4)^2+4\delta_4(1-\delta_2)}}{2\delta_4}\\
&\quad=\frac{(\alpha_1+\delta_1)(\delta_4-1)-2\delta_4}{2\delta_4}\\
&\qquad+\frac{\sqrt{[(\alpha_1+\delta_1)(\delta_4-1)-2\delta_4]^2+4\delta_4[(\alpha_1+\delta_1+\delta_4)(\alpha_1+\delta_1-1)-\delta_2(\alpha_1+\delta_1)^2]}}{2\delta_4}
\end{align}
Hence we can see that if
$$
\delta_2<\frac{(\alpha_1+\delta_1+\delta_4)(\alpha_1+\delta_1-1)}{(\alpha_1+\delta_1)^2},
$$
then $1-(\alpha_1+\delta_1)\hat{u}_{2+}<0$ and as a result SS3 with $\hat{u}_{2+}$ is stable. If
$$
\frac{(\alpha_1+\delta_1+\delta_4)(\alpha_1+\delta_1-1)}{(\alpha_1+\delta_1)^2}\le \delta_2 < \frac{(1+\delta_4)^2}{4\delta_4}\quad\text{and}\quad \frac{1}{\delta_4}<\frac{\alpha_1+\delta_1-2}{\alpha_1+\delta_1},
$$
then $1-(\alpha_1+\delta_1)\hat{u}_{2+}<0$ and as a result SS3 with $\hat{u}_{2+}$ is linearly stable.
\item[\textnormal{SS4.}]
We consider the SS solution $\bu^*=(1-(\alpha_1+\delta_1)\tilde{u}_2,\tilde{u}_2,\tilde{u}_2,[1+\delta_4\tilde{u}_2]^{-1})$ noting from $\bF(\bu)=\bzero$ we have $1-u_2-\alpha_2 u_1-\delta_2 u_4=0$ and $1-u_1-\alpha_1u_2-\delta_1u_3=0$. Hence we have the Jacobian matrix
\begin{equation}
\bF'_\bu(\bu^*)=
\begin{bmatrix}
(\alpha_1+\delta_1)\tilde{u}_2-1  &  \alpha_1[(\alpha_1+\delta_1)\tilde{u}_2-1]    &  \delta_1[(\alpha_1+\delta_1)\tilde{u}_2-1] &   0   \\
-\beta_2\alpha_2 \tilde{u}_2      &  -\beta_2\tilde{u}_2                           &  0                                        & -\beta_2\delta_2\tilde{u}_2 \\
0                               &  \beta_3                                     &  -\beta_3                                 &   0   \\
0                               &  -\beta_4 \delta_4(1+\delta_4\tilde{u}_2)^{-1} &  0                                        &   -\beta_4(1+\delta_4 \tilde{u}_2) \\
\end{bmatrix}
\end{equation}
that has a characteristic equation
\begin{equation}
\lambda^4+c_3\lambda^3+c_2\lambda^2+c_1\lambda+c_0=0,
\end{equation}
where
\begin{align}
c_3=&\,1-(\alpha_1+\delta_1)\tilde{u}_2+\beta_3 +\beta_2\tilde{u}_2+\beta_4(1+\delta_4\tilde{u}_2),\\
c_2=&\,\beta_3[\beta_2\tilde{u}_2+\beta_4(1+\delta_4\tilde{u}_2)]+[1-(\alpha_1+\delta_1)\tilde{u}_2][\beta_3+\beta_4(1+\delta_4\tilde{u}_2)+\beta_2\tilde{u}_2(1-\alpha_2\alpha_1)]\\
&\quad+\frac{\beta_2\beta_4\tilde{u}_2}{1+\delta_4\tilde{u}_2}[(1+\delta_4\tilde{u}_2)^2-\delta_2\delta_4],\\
c_1=&\,[1-(\alpha_1+\delta_1)\tilde{u}_2]\left[\frac{\beta_2\beta_4\tilde{u}_2}{1+\delta_4\tilde{u}_2}[(1+\delta_4\tilde{u}_2)^2(1-\alpha_2\alpha_1)-\delta_2\delta_4]+\beta_3\beta_4(1+\delta_4\tilde{u}_2)\right]\\
&\quad+\beta_2\beta_3\tilde{u}_2[1-(\alpha_1+\delta_1)\tilde{u}_2][1-\alpha_2(\alpha_1+\delta_1)]+\frac{\beta_2\beta_3\beta_4\tilde{u}_2}{1+\delta_4\tilde{u}_2}[(1+\delta_4\tilde{u}_2)^2-\delta_2\delta_4],\\
c_0=&\,\frac{\beta_2\beta_3\beta_4\tilde{u}_2}{1+\delta_4\tilde{u}_2}\left[1-(\alpha_1+\delta_1)\tilde{u}_2\right]\left[(1+\delta_4\tilde{u}_2)^2(1-\alpha_2(\alpha_1+\delta_1))-\delta_2\delta_4\right]
\end{align}
From Routh--Hurwitz conditions \citep[][pp.~507--509]{MurrayI2002}, $c_0,c_1,c_2,c_3>0$ and $c_1c_2c_3>c_3^2c_0+c_1^2$ if and only if $\text{Re}(\lambda)<0$. Hence we require $c_0,c_1,c_2,c_3>0$ and $c_1c_2c_3>c_3^2c_0+c_1^2$ for $\tilde{u}_2$ to be linearly stable. Note that if
$$
\tilde{u}_2>0,\quad 1-(\alpha_1+\delta_1)\tilde{u}_2>0\quad\text{and}\quad 0<\delta_2\delta_4<(1+\delta_4\tilde{u}_2)^2(1-\alpha_2(\alpha_1+\delta_1)),
$$
then $c_0>0$. Note that
$$
(1+\delta_4\tilde{u}_2)^2(1-\alpha_2(\alpha_1+\delta_1))<(1+\delta_4\tilde{u}_2)^2(1-\alpha_2\alpha_1)<(1+\delta_4\tilde{u}_2)^2,
$$
hence these conditions will imply that $c_0,c_1,c_2,c_3>0$. Since $\tilde{u}_2$ satisfies \eqref{eqn:vtilde} we have
$$
-\delta_2=(1+\delta_4\tilde{u}_2)[(1-\alpha_2(\alpha_1+\delta_1))\tilde{u}_2+\alpha_2-1].
$$
Hence we can show that
\begin{equation}\label{eqn:c0simp}
(1+\delta_4\tilde{u}_2)^2(1-\alpha_2(\alpha_1+\delta_1))-\delta_2\delta_4=(1+\delta_4\tilde{u}_2)[(1-\alpha_2(\alpha_1+\delta_1))(2\delta_4\tilde{u}_2+1)+\delta_4(\alpha_2-1)],
\end{equation}
and as a result
\begin{equation}\label{eqn:c0}
c_0=\beta_2\beta_3\beta_4\tilde{u}_2\left[1-(\alpha_1+\delta_1)\tilde{u}_2\right]\left[(1-\alpha_2(\alpha_1+\delta_1))(2\delta_4\tilde{u}_2+1)+\delta_4(\alpha_2-1)\right].
\end{equation}
Now consider \eqref{eqn:vtilde} and solve for $\tilde{u}_2$ to obtain
\begin{equation}\label{eqn:vtildepm}
\begin{split}
\tilde{u}_{2\pm}=&\,\frac{-(1-\alpha_2(\alpha_1+\delta_1))-\delta_4(\alpha_2-1)}{2\delta_4(1-\alpha_2(\alpha_1+\delta_1))}\\
&\quad\pm\frac{\sqrt{[1-\alpha_2(\alpha_1+\delta_1)+\delta_4(\alpha_2-1)]^2-4\delta_4(1-\alpha_2(\alpha_1+\delta_1))(\delta_2+\alpha_2-1)}}{2\delta_4(1-\alpha_2(\alpha_1+\delta_1))}\\
=&\,\frac{-(1-\alpha_2(\alpha_1+\delta_1))-\delta_4(\alpha_2-1)}{2\delta_4(1-\alpha_2(\alpha_1+\delta_1))}\\
&\quad\pm\frac{\sqrt{[1-\alpha_2(\alpha_1+\delta_1)-\delta_4(\alpha_2-1)]^2-4\delta_2\delta_4(1-\alpha_2(\alpha_1+\delta_1))}}{2\delta_4(1-\alpha_2(\alpha_1+\delta_1))},
\end{split}
\end{equation}
and note that $\tilde{u}_{2\pm}\in\R$ if and only if
$$
0<\delta_2\le\frac{[1-\alpha_2(\alpha_1+\delta_1)-\delta_4(\alpha_2-1)]^2}{4\delta_4[1-\alpha_2(\alpha_1+\delta_1)]}\quad\text{or}\quad 1-\alpha_2(\alpha_1+\delta_1)<0.
$$
Using \eqref{eqn:vtildepm}, we have
\begin{equation}\label{eqn:c0part}
\begin{split}
&(1-\alpha_2(\alpha_1+\delta_1))(2\delta_4\tilde{u}_{2\pm}+1)+\delta_4(\alpha_2-1)\\
&\quad=\pm\sqrt{[(1-\alpha_2(\alpha_1+\delta_1))-\delta_4(\alpha_2-1)]^2-4\delta_2\delta_4(1-\alpha_2(\alpha_1+\delta_1))}
\end{split}
\end{equation}
and therefore when $\tilde{u}_{2\pm}\in\R$ we have that \eqref{eqn:c0part} will be positive for $\tilde{u}_{2+}$ and negative for $\tilde{u}_{2-}$. We first note from \eqref{eqn:c0part} and \eqref{eqn:c0simp} that we require either $\tilde{u}_{2-}<0$ or $\tilde{u}_{1-}=1-[\alpha_1+\delta_1]\tilde{u}_{2-}<0$ so that $c_0>0$. Hence we can conclude that there will be no biologically meaningful values of SS4 with $\tilde{u}_{2-}$ that are stable. If $1-\alpha_2(\alpha_1+\delta_1)<0$, we can see that \eqref{eqn:c0simp} is negative: noting that \eqref{eqn:c0part} will be positive for $\tilde{u}_{2+}$ in this case then \eqref{eqn:c0simp} implies that $1+\delta_4\tilde{u}_{2+}<0$. Hence we can conclude that $\tilde{u}_{2+}<0$ and as a result, from \eqref{eqn:c0} we have that $c_0<0$ if $1-\alpha_2(\alpha_1+\delta_1)<0$ for $\tilde{u}_{2+}$. Therefore if $1-\alpha_2(\alpha_1+\delta_1)<0$, then SS4 with $\tilde{u}_{2+}$ is unstable. If $1-\alpha_2(\alpha_1+\delta_1)>0$ and $\tilde{u}_{2+}<0$, then we can see from \eqref{eqn:c0part} that \eqref{eqn:c0} will be negative for $\tilde{u}_{2+}$ and hence SS4 with $\tilde{u}_{2+}$ will be unstable. Therefore we can see that for SS4 with $\tilde{u}_{2+}$ to be stable it is necessary that
\begin{equation}
\tilde{u}_{2+}>0, \quad 1-(\alpha_1+\delta_1)\tilde{u}_{2+}>0 \quad \text{and}\quad 0<\delta_2<\frac{[1-\alpha_2(\alpha_1+\delta_1)-\delta_4(\alpha_2-1)]^2}{4\delta_4[1-\alpha_2(\alpha_1+\delta_1)]}.
\end{equation}
\end{enumerate}
\hfill\end{proof}
\end{appendix}

\section*{Acknowledgements}
ABH has been supported by an Australian Postgraduate Award.

\bibliographystyle{plainnat}
\bibliography{AMTI-CI-Bibliography}

\end{document}